\documentclass[12pt]{article}
\usepackage{amsmath, amssymb, amsthm}
\usepackage{algorithm, algorithmic}
\usepackage{setspace}
\usepackage{geometry}
\usepackage{graphicx}
\usepackage{url}
\usepackage{hyperref}
\usepackage{booktabs}
\usepackage{threeparttable}

\usepackage[suppress]{color-edits}
\usepackage{tikz}
\usepackage{caption}
\usepackage{subcaption}
\usetikzlibrary{decorations.pathmorphing, arrows.meta, calc}
\usetikzlibrary{decorations.pathreplacing}

\newtheorem{theorem}{Theorem}
\newtheorem{lemma}[theorem]{Lemma}
\newtheorem{corollary}[theorem]{Corollary}

\newtheorem{definition}[theorem]{Definition}
\newtheorem{claim}[theorem]{Claim}

\newtheorem*{openq}{Open Question}
\newtheorem*{itheorem}{Informal Theorem}

\usepackage{tcolorbox}
\renewenvironment{openq}{%
    \begin{tcolorbox}[colback=gray!7, colframe=black, boxsep=1pt, left=4pt, right=4pt, rounded corners]%
    \textbf{Open Question:}
}{%
    \end{tcolorbox}%
}

\newtheorem{problem}[]{Problem}

\renewcommand{\hat}{\widehat}
\renewcommand{\tilde}{\widetilde}
\renewcommand{\bar}{\overline}

\DeclareMathOperator{\poly}{poly}

\DeclareMathOperator{\OPT}{OPT}

\def\min{\qopname\relax n{min}}
\def\max{\qopname\relax n{max}}

\def\E{\mathbb{E}}
\def\EE{\mathcal{E}}

\def\w{\vec w}

\def\eps{\epsilon}

\newcommand{\ALG}{\mathrm{ALG}}

\newcommand{\AlternativeExists}{Lemma~\ref{lem:missed-properties} (\textsc{AlternativeExists})}

\title{Near-Optimal Sparsifiers for Stochastic Knapsack and Assignment Problems}
\author{
  Shaddin Dughmi\thanks{University of Southern California. Email: \texttt{shaddin@usc.edu}.}
  \and
  Yusuf Hakan Kalayci\thanks{University of Southern California. Email: \texttt{kalayci@usc.edu}.}
  \and
  Xinyu Liu\thanks{University of Southern California. Email: \texttt{xinyul@usc.edu}.}
}

\date{September 2025}

\begin{document}

\maketitle

\begin{abstract}
When uncertainty meets costly information gathering, a fundamental question emerges: which data points should we probe to unlock near-optimal solutions? Sparsification of stochastic packing problems addresses this trade-off. 
The existing notions of sparsification measure the level of sparsity, called degree, as the ratio of queried items to the optimal solution size.
While effective for matching and matroid-type problems with uniform structures, this cardinality-based approach fails for knapsack-type constraints where feasible sets exhibit dramatic structural variation.
We introduce a polyhedral sparsification framework that measures the degree as the smallest scalar needed to embed the query set within a scaled feasibility polytope, naturally capturing redundancy without relying on cardinality.

Our main contribution establishes that knapsack, multiple knapsack, and generalized assignment problems admit $(1-\epsilon)$-approximate sparsifiers with degree polynomial in $1/p$ and $1/\epsilon$ -- where $p$ denotes the independent activation probability of each element -- remarkably independent of problem dimensions.
The key insight involves grouping items with similar weights and deploying a charging argument: when our query set misses an optimal item, we either substitute it directly with a queried item from the same group or leverage that group's excess contribution to compensate for the loss. 
This reveals an intriguing complexity-theoretic separation -- while the multiple knapsack problem lacks an FPTAS and generalized assignment is APX-hard, their sparsification counterparts admit efficient $(1-\epsilon)$-approximation algorithms that identify polynomial degree query sets.
Finally, we raise an open question: can such sparsification extend to general integer linear programs with degree independent of problem dimensions?
\end{abstract}




\section{Introduction}

The sparsification of stochastic packing problems has emerged as a fundamental paradigm for designing algorithms that achieve near-optimal solutions with limited access to uncertain data. This approach proves particularly valuable in settings where probing or querying elements incurs costs or faces constraints. Recent developments in sparsification of stochastic packing problems have revealed elegant techniques for selecting the most pertinent information to probe. 

Our goal in this paper is to expand the sparsification toolbox beyond well-studied matching settings~\cite{assadi2019towards, azarmehr2025, behnezhad2022stochastic, behnezhad2020weighted, behnezhad2020stochastic, behnezhad2019stochastic, blum2015ignorance,derakhshan2023stochastic,derakhshan2025weighted, yamaguchi2018stochastic} and matroid optimization problems~\cite{dughmi2023sparsification,maehara2019stochastic, yamaguchi2018stochastic} to design sparsifiers for knapsack-type constraints. A detailed survey of related work appears later in Section~\ref{sec:related_works}.

We adopt the general framework of Dughmi et al. \cite{dughmi2023sparsification} for a packing problem instance $\langle E, \mathcal{F}, f\rangle$, where $E$ is a ground set of elements, $\mathcal{F} \subseteq 2^E$ is a downward-closed family of feasible sets, and $f : 2^E \to \mathbb{R}_{\geq 0}$ is an objective function. In our stochastic setting, each element $e \in E$ becomes \emph{active} independently with probability $p \in (0,1]$, resulting in a random active set $R \subseteq E$. A \emph{sparsification algorithm} selects a query set $Q \subseteq E$ prior to observing $R$, aiming to maximize the solution value within the revealed subset $Q \cap R$.

For assignment problems like MKP and GAP, solutions are typically defined as sets of item-knapsack pairs (assignments) rather than simple subsets of items. To align with our sparsification framework, we project these problems onto the ground set of items $E$. We say a subset of items $S \subseteq E$ is \emph{feasible} ($S \in \mathcal{F}$) if there exists a valid assignment of items in $S$ to knapsacks that respects all capacity constraints. Accordingly, we extend the objective function to item sets by defining $f(S)$ as the maximum value achievable by any feasible assignment of the items in $S$.

To rigorously evaluate sparsification algorithms, we must balance approximation quality against the ``size'' of the query set. While approximation is defined as usual, quantifying the size of $Q$ for knapsack-type problems requires nuance. Traditional cardinality-based measures -- which normalize $|Q|$ by the rank of $\mathcal{F}$ -- fail when feasible sets vary dramatically in size, as in capacity-constrained problems. To address this, we introduce the \emph{polyhedral sparsification degree}. Let $\mathcal{P}_\mathcal{F} = \operatorname{conv}(\{\mathbf{1}_S : S \in \mathcal{F}\})$ be the polytope of feasible solutions. We define the degree of a query set $Q$ as the minimum scalar $d \geq 1$ such that the normalized indicator vector $\frac{1}{d}\mathbf{1}_Q$ lies within $\mathcal{P}_\mathcal{F}$. This definition naturally generalizes existing notions\footnote{Sparsification in stochastic integer packing was introduced by Blum et al. \cite{blum2015ignorance} using vertex degree, and later generalized by Maehara and Yamaguchi \cite{maehara2019stochastic, yamaguchi2018stochastic} to cardinality-based measures.} by capturing the intuitive notion of redundancy: a query set has degree $d$ if it can be decomposed into $d$ (fractionally) feasible solutions.

\begin{definition}[Sparsifier]
    An algorithm $\mathcal{A}$ is an \emph{$\alpha$-approximate sparsifier with degree $d$} if, for every problem instance, it returns a (possibly randomized) query set $Q \subseteq E$ that satisfies two conditions: 
    \begin{enumerate}
        \item \emph{Polyhedral Feasibility:} the indicator vector of the query set lies within the polytope scaled by $d$, in other words $\frac{1}{d} \cdot \mathbf{1}_Q \in \mathcal{P}_\mathcal{F}$ (holding for every realization of $Q$);
        \item \emph{Approximation Guarantee:} the expected value of the optimal solution within the queried active elements approximates the full-information optimum, such that 
        $$\E_{R, Q}[\max \{ f(S) : S \in \mathcal{F} \cap 2^{Q \cap R} \}] \geq \alpha \cdot \E_R[\OPT].$$
    \end{enumerate}
\end{definition}

\paragraph{Our contributions.}
Our main result establishes that knapsack-type problems admit effective sparsification despite their computational hardness. We design deterministic and non-adaptive
algorithms that produce $(1-\epsilon)$-approximate sparsifiers with polyhedral degree polynomial in $1/p$ and $1/\epsilon$, remarkably independent of the number of items or constraints. Unlike some prior work \cite{maehara2019stochastic, yamaguchi2018stochastic}, we require our degree to be independent of the total number of variables and constraints.\footnote{To be clear, our definition of degree implicitly allows the number of queries to scale with the cardinality of a typical solution. Our pursuit of constant degree (independent of the number of items or constraints) is akin to requiring ``redundancy'' on the order of a constant number of solutions, as a hedge against uncertainty.}

\begin{itheorem}[Main Result]
\label{thm:informal_main}
For parameters $\epsilon \in (0,1/6)$ and $p \in (0,1]$, there exist efficient algorithms that produce a $(1-O(\epsilon))$-approximate sparsifier for the stochastic Knapsack, Multiple Knapsack, and Generalized Assignment Problems with sparsification degree $\text{poly}(1/\epsilon, 1/p)$.
\end{itheorem}

For the single knapsack problem, our approach employs a bucketization strategy combined with a charging argument. Items are partitioned into buckets based on value, and each bucket is filled to approximately $\text{poly}(1/p, 1/\epsilon)$ times the knapsack capacity. This redundancy allows for a substitution mechanism: if an optimal item is not queried, it can likely be replaced by a queried item from the same bucket. Feasibility is maintained by prioritizing items with smaller weights, with a refined density-based strategy for the smallest value bucket.

Extending this to the Generalized Assignment Problem (GAP) presents a fundamental obstacle: item characteristics are knapsack-dependent, breaking the direct substitutability essential to the single knapsack case. An item ideal for one knapsack may be highly inefficient for another. We overcome this by increasing redundancy -- filling buckets to $\text{poly}(1/p, 1/\epsilon) \cdot (1/\epsilon)$ capacity -- and developing a sophisticated charging argument. When an optimal item cannot be directly substituted (often because potential substitutes are assigned elsewhere), we fractionally distribute its lost value across multiple queried elements. This ensures that we recover nearly the full optimal value without overburdening any specific element in the analysis. 


\paragraph{Implications and experiments.}
We turn to the deterministic setting (i.e., $p = 1$) to reflect on complexity-theoretic implications and practical impact. Our results reveal an intriguing separation: while GAP is APX-hard and MKP lacks an FPTAS, their stochastic counterparts admit efficient $(1-\epsilon)$-approximate sparsifiers with polynomial degree. From the Exponential Time Hypothesis \cite{eth}, one would informally expect that the “hard’’ instances are already sparse, whereas “easy’’ instances may be more dense and amenable to sparsification. It is on those dense non–worst-case instances where our sparsifiers shrink the search space. This can serve as a useful preprocessing step for heuristic algorithms, guiding them to pertinent variables, even in the deterministic setting with $p=1$. 

Empirically, we validate the utility of our approach on synthetic datasets, using practical choices of hyperparameters $(\alpha, \tau, \epsilon, K)$ that are less conservative than the theoretical settings needed to ensure worst-case, high-probability guarantees. These practical choices result in sparsifiers with significantly smaller degree than the theoretical degree guaranteed by our theorem, and on instances where the total number of items exceeds the optimal solution size by a factor of four, our sparsification algorithm reduces runtime by $4\times$ while preserving $99\%$ of the solution quality. Furthermore, under fixed time budgets, branch-and-bound algorithms running on our sparsified instances outperform those running on full datasets by a factor of $5$ in objective value. 

\paragraph{Future directions.}
Finally, our work motivates a broader question regarding integer linear programs (ILPs). Current sparsifiers for general ILPs \cite{yamaguchi2018stochastic} depend on problem dimensions or column sparsity. This leaves open a fundamental challenge:
\begin{openq}
    Can we design $(1-\epsilon)$-approximate sparsifiers for general integer linear programs where the sparsification degree scales polynomially with $1/p$, $1/\epsilon$, and intrinsic structural parameters, but remains independent of the total number of variables and constraints? 
\end{openq}
\noindent Resolving this would significantly advance the understanding of information requirements in stochastic optimization.

\section{Stochastic Assignment Problems}

In this section, we formally define the deterministic assignment problems addressed in this work -- Knapsack, Multiple Knapsack, and Generalized Assignment -- and their stochastic counterparts. Throughout our discussion, we use the terms ``knapsack'' and ``bin'' interchangeably.

\subsection{Problem Definitions}

We begin with the classical single-bin setting and progressively generalize to multiple heterogeneous constraints.

\begin{problem}[Knapsack Problem (KP01)]
    Consider a single knapsack with capacity $C$ and a collection of $n$ items denoted by $E$. Each item $i \in E$ is characterized by a value $v_i$ and a weight $w_i$. The objective is to select a subset of items that maximizes total value while respecting the capacity constraint:
    \[
    \max \sum_{i \in S} v_i \quad \text{subject to} \quad \sum_{i \in S} w_i \leq C, \quad S \subseteq [n].
    \]  
\end{problem}

\begin{problem}[Multiple Knapsack Problem (MKP)]
    Consider $m$ knapsacks where knapsack $j \in [m]$ has capacity $C_j$, and $n$ items $E$ where each item $i \in E$ has value $v_i$ and weight $w_i$. The objective is to assign items to the knapsacks to maximize total value while ensuring no knapsack exceeds its capacity. Formally, we seek disjoint subsets $S_j \subseteq E$ for $j \in [m]$ such that:
    \[
    \max \sum_{j=1}^m \sum_{i \in S_j} v_i \quad \text{subject to} \quad \sum_{i \in S_j} w_i \leq C_j \quad \text{for all } j \in [m].
    \]
\end{problem}

\begin{problem}[Generalized Assignment Problem (GAP)]
    Consider $m$ knapsacks where knapsack $j \in [m]$ has capacity $C_j$, and $n$ items $E$. In contrast to the previous problems, each item $i \in E$ exhibits knapsack-dependent characteristics: when assigned to knapsack $j$, item $i$ contributes value $v_{ij}$ and consumes weight $w_{ij}$. The objective is to find disjoint subsets $S_j \subseteq E$ for $j \in [m]$ that maximize total value subject to capacity constraints:
    \[
    \max \sum_{j=1}^m \sum_{i \in S_j} v_{ij} \quad \text{subject to} \quad \sum_{i \in S_j} w_{ij} \leq C_j \quad \text{for all } j \in [m].
    \]
\end{problem}
\paragraph{Stochastic Variants.}
For any deterministic instance $\Pi$ defined above, its stochastic variant $\langle \Pi, p\rangle$ is characterized by a parameter $p \in (0,1]$. A random subset $R \subseteq E$, termed the \emph{active set}, is generated by sampling each element $e \in E$ independently with probability $p$. The goal is to select a feasible solution using only the items present in the realization $R$ to maximize the objective value.

\paragraph{Assumptions.}
Without loss of generality, we assume that every individual item is feasible. That is, for every item $i \in E$, there exists at least one knapsack $j$ such that the item's weight fits within the capacity ($w_{ij} \leq C_j$).

\subsection{Notation and Feasibility}

We distinguish between items and assignments using non-bold and bold notation, respectively.

\begin{itemize}
    \item \textbf{Assignments ($\mathbf{S}$):} We denote a specific assignment as $\mathbf{S} \subseteq E \times [m]$, where pairs $(i,j) \in \mathbf{S}$ indicate that item $i$ is assigned to knapsack $j$. For a solution to be valid, each item must appear in at most one pair. We define the total value $v(\mathbf{S}) = \sum_{(i,j) \in \mathbf{S}} v_{ij}$ and the weight consumed in knapsack $j$ as $w_j(\mathbf{S}) = \sum_{(i,j) \in \mathbf{S}} w_{ij}$.
    \item \textbf{Item Sets ($S$):} When $S \subseteq E$ denotes a subset of the ground set, it refers purely to the items themselves. We abuse notation slightly in the context of GAP: for a subset of items $S$ implicitly assigned to a specific knapsack $j$, we write $w_j(S) = \sum_{i \in S} w_{ij}$.
\end{itemize}

A crucial distinction in our framework is the definition of feasibility for item sets. While the optimization problems maximize over assignments $\mathbf{S}$, our sparsification framework operates on the ground set $E$. 
We say that an item set $S \subseteq E$ is \emph{feasible} if there exists a valid assignment $\mathbf{S}$ such that the set of assigned items is exactly $S$ (i.e., $S = \{i \mid \exists j, (i,j) \in \mathbf{S}\}$).
Consequently, the feasibility family $\mathcal{F}$ used to define the polytope $\mathcal{P}_{\mathcal{F}}$ consists of all such feasible item sets. 
Therefore, the condition for sparsification degree, $\mathbf{1}_Q \in d \cdot \mathcal{P}_{\mathcal{F}}$, is a constraint on the query set $Q$ in the item space. The specific assignment $\mathbf{S}$ is determined only after the intersection of the query set and active set, $Q \cap R$, is revealed.

\subsection{Hardness and Approximability}
The computational complexity of these problems forms a natural hierarchy. The classical knapsack problem, which was proven NP-hard by Karp \cite{karp1972reducibility}, admits a fully polynomial-time approximation scheme (FPTAS)~\cite{chen2024nearly, ibarra1975fptas, mao2024approximation}. However, the Multiple Knapsack Problem (MKP), even when restricted to just two knapsacks, does not admit an FPTAS~\cite{chekuri2005mkp}. The Generalized Assignment Problem (GAP) is APX-hard; Chakrabarty and Goel \cite{chakrabarty2010gapbound} demonstrated that it cannot be approximated better than a factor of $10/11$ unless $\mathrm{P}=\mathrm{NP}$. The current best polynomial-time approximation algorithm for GAP achieves a ratio of $1 - 1/e + \varepsilon$ for a small constant $\varepsilon > 0$~\cite{feige2006allocation}.

\section{Warm-up: Knapsack Sparsification}
We begin with a sparsification algorithm for the classical knapsack problem. This foundational case introduces key techniques that will be extended to more complex scenarios throughout the paper.

The algorithm employs a bucket-based strategy that partitions elements by value into geometrically increasing ranges. The fundamental principle ensures that the queried subset of items in each bucket is sufficiently large to either contain all items within the bucket or independently fill the knapsack constraint with high probability. In the former case, all elements remain available in the query set, ensuring that no item from the optimal solution is missed. In the latter scenario, the objective is to guarantee that whenever an item from the optimal solution is unavailable in the query set, a suitable substitute can be found.

To achieve this, the algorithm applies distinct selection criteria: for low-value elements (bucket $B_0$), it prioritizes items with high value-to-weight density, while for high-value buckets ($B_k$ with $k \geq 1$), it selects the lightest elements to maximize the probability of accommodating valuable items within capacity constraints.

\begin{algorithm}[h]
    \caption{Bucket-Based Sparsifier for Knapsack}
    \label{alg:bucket_KP_sparsifier}
    \begin{algorithmic}[1]
        \STATE \textbf{Input:} accuracy parameter $\epsilon \in (0,1/3)$.
        \STATE Set concentration factor $\tau(\eps) = 1 + \ln\left(\frac{1}{\eps}\right) + \sqrt{ \ln^2\left(\frac{1}{\eps}\right) + 2 \ln\left(\frac{1}{\eps}\right) }$ as defined in Lemma~\ref{lem:bucket_weight_guarantee}.
        \STATE Estimate $\E_R[\OPT]$ via sampling to obtain $(1 \pm \epsilon)$-approximation $M$.
        \STATE Initialize query set \( Q \gets \emptyset \).
        \STATE Set number of buckets $K := \lceil \frac{1}{\epsilon} \log(\frac{1}{\epsilon p}) \rceil$.
        \STATE Partition elements into buckets:
        $$B_k = \begin{cases}
            \{i \in E : v_i \leq \epsilon M\} & \text{if } k = 0 \\
            \{i \in E : v_i \in (\epsilon(1+\epsilon)^{k-1} M, \epsilon(1+\epsilon)^k M]\} & \text{if } 1 \le k \le K
        \end{cases}$$
        \STATE Sort $B_0$ by decreasing value density $v_i/w_i$.
        \STATE Define $\overline{B}_0$ as the shortest prefix of $B_0$ with total weight $\sum_{i \in \overline{B}_0} w_i \geq \frac{\tau(\epsilon)}{p} C$, or set $\overline{B}_0 = B_0$ if no such prefix exists.
        \STATE Add to query set: $Q \leftarrow Q \cup \overline{B}_0$.
        
        \FOR{\( k = 1 \) to \( K \)}
            \STATE Sort \( B_k \) in ascending order of weight \( w_i \).
            \STATE Let \( \overline{B}_k \) be the minimal prefix of $B_k$ such that
            \(
            \sum_{i \in \overline{B}_k} w_i \geq \frac{\tau(\epsilon)}{p} \cdot C
            ,
            \)
            or let \( \overline{B}_k := B_k \) if no such prefix exists.
            \STATE Update \( Q \gets Q \cup \overline{B}_k \).
        \ENDFOR
        
        \STATE \textbf{Output:} Return the query set \( Q \).
    \end{algorithmic}
\end{algorithm}

Before we proceed with proving that Algorithm~\ref{alg:bucket_KP_sparsifier} is a good sparsifier, we first establish a key concentration result for our sparsifier analysis. The proof follows from standard concentration inequalities and is presented in Appendix~\ref{app:weight-concentration}.

\begin{lemma}[Activation Weight Concentration] \label{lem:bucket_weight_guarantee}
Let \( S \subseteq E \) be a set of elements, each with weight \( w_i \), such that
\[
\sum_{i \in S} w_i \geq \frac{\tau(\eps)}{p} \cdot C,
\]
where
\(
\tau(\eps) := 1 + \ln(1/\eps) + \sqrt{ \ln^2(1/\eps) + 2 \ln(1/\eps) }.
\)
Then, if each item is active independently with probability \( p \), the total weight of active items in \( S \) is at least \( C \) with probability at least \( 1 - \eps \). 
\end{lemma}

\begin{theorem}[Knapsack Sparsifier Performance]
\label{thm:knapsack_sparsifier}
For parameters $\epsilon \in (0,1/3)$ and $p \in (0,1]$, Algorithm~\ref{alg:bucket_KP_sparsifier} produces a $(1-4\epsilon)$-approximate sparsifier for the stochastic knapsack problem with sparsification degree
\[
O\left(\frac{\log(1/\epsilon)\cdot \log(1/(\epsilon p))}{\epsilon p}\right).
\]

\end{theorem}

\begin{proof}
Let $M$ denote our $(1 \pm \epsilon)$-approximation to $\E_R[\OPT]$, satisfying $(1-\epsilon) \cdot \E_R[\OPT] \leq M \leq (1+\epsilon) \cdot \E_R[\OPT]$ with probability at least $1-\epsilon$. Let $\EE_{\text{est}}$ be the event that this inequality holds and so $M$ is estimated within $(1\pm\epsilon)$ approximation. Condition on the event $\EE_{\text{est}}$ holds.

\paragraph{Sparsification Degree Analysis.}
We first bound the size of the query set $Q$. The algorithm selects at most $K+1$ buckets. For each bucket $\overline{B}_k$, the total weight is bounded by $\sum_{i \in \overline{B}_k} w_i \leq C \cdot \frac{\tau(\epsilon)}{p} + \max_{i} w_i \leq C \left( \frac{\tau(\epsilon)}{p} + 1 \right)$. Summing over all $K = O(\frac{1}{\epsilon} \log(\frac{1}{\epsilon p}))$ buckets, the total weight of the query set satisfies:
\[
w(Q) = \sum_{k=0}^K w(\overline{B}_k) \leq O\left( \frac{\tau(\epsilon)}{p} \cdot K \right) \cdot C.
\]
To translate this weight bound into our polyhedral sparsification degree, we consider the linear programming relaxation of the knapsack polytope, $\mathcal{P}_{LP} = \{ x \in [0,1]^E \mid \sum x_i w_i \leq C \}$. Our calculation shows that $\mathbf{1}_Q \in d' \cdot \mathcal{P}_{LP}$ for $d' = O(\frac{w(Q)}{C})$. However, the sparsification degree requires embedding into the convex hull of integer solutions, $\mathcal{P}_{\mathcal{F}}$. We know that the integrality gap of the knapsack relaxation is bounded by 2 (assuming singletons are feasible) \cite{williamson2011design}. Therefore, $\mathcal{P}_{LP} \subseteq 2 \cdot \mathcal{P}_{\mathcal{F}}$, implying that the sparsification degree is at most $2d'$, which remains $O\left(\frac{\log(1/\epsilon)\cdot \log(1/(\epsilon p))}{\epsilon p}\right)$.

\paragraph{Approximation Analysis.}
Consider any realization $R \subseteq E$ and let $S^* \subseteq R$ denote an optimal solution with value $\OPT(R) = \sum_{i \in S^*} v_i$ and weight $w(S^*) \leq C$.

We partition the optimal solution as $S^* = S_0 \cup S_1 \cup \cdots \cup S_K$ where $S_k = S^* \cap B_k$, and define $S^{\text{low}} = S_0$ (low-value items) and $S^{\text{high}} = \bigcup_{k=1}^K S_k$ (high-value items).
Completeness of this partition follows from the range of buckets. Since each item $i$ is active with probability $p$, $\E[\OPT] \geq p v_i$, implying $v_i \leq \E[\OPT]/p$. Our largest bucket boundary is at least $M/p = \E[\OPT]/p$, ensuring all items are covered.

\emph{High-Value Item Recovery.} For each bucket $k \geq 1$, let $Q_k = \overline{B}_k \cap R$ represent the active queried items. By Lemma~\ref{lem:bucket_weight_guarantee}, we have $w(Q_k) \geq C$ with the probability of at least $1-\epsilon$ when $w(\bar{B_k})$ is sufficiently large (equivalently when $\overline{B}_k \neq B_k$). Define the event $\EE_k$ as $\overline{B}_k = B_k$ or $w(Q_k) \geq C$.

Conditioning on this event $\EE_k$, since $\overline{B}_k$ contains the lightest items in $B_k$, we can establish a matching $\phi_k: S_k \rightarrow Q_k$ such that each item $i \in S_k$ maps to some $\phi_k(i) \in Q_k$ with $w_i \geq w_{\phi_k(i)}$ and $v_i \leq (1+\epsilon) \cdot v_{\phi_k(i)}$ (due to items being in the same value bucket). Notice that such a matching trivially exists when $\overline{B}_k = B_k$.
This matching implies:
\begin{align*}
    \E_R\left[ \max_{\substack{T_k \subseteq Q_k\\w(T_k) \leq w(S_k)}} v(T_k) \mid \EE_k \right] \geq \frac{1}{1+\epsilon} \cdot \E_R[v(S_k)].
\end{align*}
Since $\Pr[\EE_k] \geq (1-\epsilon)$, we obtain:
\begin{equation}
    \label{eq:high_value_bound}
    \E_R\left[ \max_{\substack{T_k \subseteq Q_k\\w(T_k) \leq w(S_k)}} v(T_k) \right] \geq (1-\epsilon) \cdot \frac{1}{1+\epsilon} \cdot \E_R[v(S_k)] \geq (1-2\epsilon) \cdot \E_R[v(S_k)],
\end{equation}
where the final inequality uses $1/(1+\epsilon) \geq 1-\epsilon$ for small $\epsilon$.

\emph{Low-Value Item Recovery.} For bucket $B_0$, we apply greedy selection on $\overline{B}_0 \cap R$ by value density up to a total capacity of $w(S_0) := \sum_{i \in S_0} w_i$. Let $T_0$ denote this greedy solution.

Each item in $B_0$ has value at most $\epsilon M$. When $\EE_{\text{est}}$ holds, the maximum value in $B_0$ is at most $\epsilon(1+\epsilon) \E_R[\OPT] \leq 2\epsilon \E_R[\OPT]$. By fractional knapsack analysis, the greedy algorithm achieves value within $2\epsilon \E_R[\OPT]$ of the fractional optimum:
$$\E_R[v(T_0) \mid \EE_{\text{est}}] \geq \E_R[v(S_0) \mid \EE_{\text{est}}] - 2\epsilon \E_R[\OPT].$$
Since $\Pr[\EE_{\text{est}}] \geq 1-\epsilon$:
$$\E_R[v(T_0)] \geq (1-\epsilon) \cdot \E_R[v(S_0)] - 2 \epsilon \cdot \E_R[\OPT].$$

\emph{Final Approximation Bound.} Combining our bounds for high-value and low-value items, let $T_k$ denote the maximum-value feasible subset of $Q_k$ with \( w(T_k) \leq w(S_k) \) for each $k \geq 1$, and define $T = \bigcup_{k=0}^K T_k$. Then:
\begin{align*}
    \E_R[v(T)] 
        &= \sum_{k=0}^K \E_R[v(T_k)]\\
        &\geq (1-\epsilon) \cdot \E_R[v(S_0)] - 2\epsilon \cdot \E_R[\OPT] + (1-2\epsilon) \cdot \sum_{k=1}^K \E_R[v(S_k)]\\
        &= (1-2\epsilon) \cdot \E_R[v(S^*)] - 2\epsilon \cdot \E_R[\OPT]\\
        &\geq (1-4\epsilon) \cdot \E_R[\OPT].
\end{align*}

Finally, since $w(T_k) \leq w(S_k)$ for each $k$, we have $w(T) = \sum_{k=0}^K w(T_k) \leq \sum_{k=0}^K w(S_k)$ $= w(S^*) \leq C$, so $T$ is feasible.

\end{proof}

\section{Sparsifier for the General Assignment Problem}

We now extend our sparsification framework to the Generalized Assignment Problem (GAP), which encompasses the Multiple Knapsack Problem as a special case.

\subsection{Key Challenges and Algorithmic Innovations}

Extending our knapsack sparsifier to the GAP presents fundamental challenges that require a complete algorithmic redesign. The core difficulty stems from knapsack-dependent item characteristics, which destroy the substitutability properties essential to our knapsack analysis.

In the knapsack problem, items have fixed values $v_i$ and weights $w_i$, enabling a global bucketing scheme where items with similar characteristics substitute seamlessly. GAP breaks this structure: items exhibit knapsack-specific parameters $(v_{ij}, w_{ij})$, so an item valuable for one knapsack may be worthless for another. This forces us to maintain separate buckets $B_{j,k}$ for each knapsack-bucket pair, immediately complicating the design.

The main challenge arises from cross-knapsack substitutability issue. Two items may belong to the same bucket for knapsack $j$ due to similar values $v_{ij}$, yet reside in different buckets for knapsack $j'$ due to vastly different values $v_{ij'}$. When our sparsifier fills bucket $B_{j,k}$ based on suitability for knapsack $j$, these items fail as substitutes if the optimal solution assigns corresponding items to different knapsacks. This dependency fundamentally breaks the matching argument underlying our knapsack analysis.

We address this through enhanced redundancy combined with a more involved charging argument. Our GAP sparsifier fills each bucket to $\poly(1/\epsilon)$ times the knapsack capacity, creating substantial redundancy in the query set. When an optimal item $i^*$ assigned to knapsack $j^*$ is missing from the query set, we first seek a direct substitute among queried items with similar weight and value characteristics for knapsack $j^*$. When direct substitution fails -- typically because suitable substitutes are assigned to different knapsacks in the optimal solution -- we leverage the redundancy to fractionally charge value $v_{i^*j^*}$ across $\textrm{poly}(1/\epsilon)$ other queried items. This ensures no queried item receives excessive charge (at most $1 + \textrm{poly}(\epsilon)$ times its own value) while recovering nearly the optimal value.

A secondary challenge is ensuring the completeness of the bucket structure. In the knapsack setting,  $\mathbb{E}_R[\mathrm{OPT}]/p$ provides a natural upper bound for feasible item values. In GAP, however, an item $i$ may be active with probability $p$, yet appear in a specific knapsack $j$'s optimal assignment with significantly lower probability, making localized value bounds difficult to establish.

To address this, we introduce a \emph{super bucket} for each knapsack with no upper bound on its value range. While items in this bucket may have arbitrarily large values and lack mutual substitutability, we prove a surprising property: even if the reconstruction algorithm makes \emph{no attempt} to substitute for missed items in the super bucket, the aggregate loss remains globally bounded.

\subsection{Algorithm Design}

Our GAP sparsifier employs a bucket-based approach that incorporates substantial redundancy to handle cross-knapsack dependencies. The complete procedure is presented in Algorithm~\ref{alg:query_set_construction}.

\begin{algorithm}[!htb]
\caption{Bucket-Based Sparsifier for GAP}
\label{alg:query_set_construction}
\begin{algorithmic}[1]
\STATE \textbf{Oracle access:} For each knapsack \( j \in [m] \), assume oracle access to \( \E_R[\mathrm{OPT}_j] \), and denote it by \( M_j \).
\STATE \textbf{Input:} accuracy parameter $\epsilon \in (0,1/3)$.
\STATE Set concentration factor $\tau(\eps) = 1 + \ln\left(\frac{1}{\eps}\right) + \sqrt{ \ln^2\left(\frac{1}{\eps}\right) + 2 \ln\left(\frac{1}{\eps}\right) }$ as defined in Lemma~\ref{lem:bucket_weight_guarantee}.
\STATE Initialize query set \( Q \gets \emptyset \).
\STATE Define the number of buckets \( K := \left\lceil \frac{2}{\epsilon^2} \log \left( \frac{1}{\epsilon^3} \right) \right\rceil \).

\STATE \textbf{Bucket Formation:}
\FOR{each knapsack \( j \in [m] \)}
    \STATE Set buckets:
    \[
    B_{j,k} := 
    \begin{cases}
        \{ (i,j) \mid i \in E,\ v_{ij} \le \epsilon^2 M_j \}, & k = 0, \\[6pt]
        \{ (i,j) \mid i \in E,\ v_{ij} \in (\epsilon^2 (1+\epsilon^2)^{k-1} M_j,\ \epsilon^2 (1+\epsilon^2)^k M_j] \}, & 1 \le k \le K, \\[6pt]
        \{ (i,j) \mid i \in E,\ v_{ij} > \epsilon^2 (1+\epsilon^2)^K M_j \ \ge\ M_j / \epsilon \}, & k = K+1.
    \end{cases}
    \]
\ENDFOR

\STATE \textbf{Preprocessing:} For each \( (i,j) \), define \( \beta(i,j) \) as the bucket index such that \( (i,j) \in B_{j, \beta(i,j)} \).

\STATE \textbf{Iterative Assignment:}
\STATE Define the number of rounds \( \alpha := \frac{1}{\epsilon} \).
\FOR{round \( t = 1 \) to \( \alpha - 1 \)} 
    \FOR{each knapsack \( j \in [m] \), each bucket \( k \in \{0,1,\dots,K+1\} \)}
        \STATE Initialize remaining capacity: \( b(j,k) \gets \frac{\tau(\epsilon^2)}{p} \cdot C_j\). 
    \ENDFOR
    
    \WHILE{there exists \( (i,j) \) with \( i \notin Q \), \( \beta(i,j) \geq 1 \), \( w_{ij} \leq C_j \), and \( b(j, \beta(i,j)) > 0 \)}
        \STATE Select \( (i,j) \) with minimal \( w_{ij} \) among valid pairs.
        \STATE Update \( Q \gets Q \cup \{ i \} \) and remaining capacity \( b(j,\beta(i,j)) \gets b(j,\beta(i,j)) - w_{ij} \).
    \ENDWHILE

    \WHILE{there exists \( (i,j) \) with \( i \notin Q \), \( \beta(i,j) = 0 \), \( w_{ij} \leq C_j \), and \( b(j,0) > 0 \)}
        \STATE Select \( (i,j) \) maximizing \( v_{ij}/w_{ij} \) among valid pairs.
        \STATE Update \( Q \gets Q \cup \{ i \} \) and remaining capacity \( b(j,0) \gets b(j,0) - w_{ij} \).
    \ENDWHILE
\ENDFOR

\STATE \textbf{Output:} Return the constructed query set \( Q \).
\end{algorithmic}
\end{algorithm}

Beyond the algorithmic complexities discussed above, our GAP sparsifier requires access to knapsack-level value estimates. This presents an additional challenge compared to the knapsack setting, where estimating the global expectation $\E_R[\mathrm{OPT}]$ suffices for bucket boundary determination. In GAP, assignment decisions are interdependent across knapsacks, creating a more complex estimation problem.

Formally, for each realization $R$ of the active set, let $\mathcal{OPT}(R)$ denote the set of all optimal feasible GAP assignments on $R$. We fix once and for all an arbitrary but deterministic tie-breaking rule that selects a canonical optimal assignment $\mathbf{OPT}(R) \in \mathcal{OPT}(R)$. We then define $\mathrm{OPT}(R)$ as the total value of assignment $\mathbf{OPT}(R)$ and $\mathrm{OPT}_j(R)$ as the total value of items assigned to knapsack $j$ in $\mathbf{OPT}(R)$, so that
\[
\mathrm{OPT}(R) = \sum_{j=1}^m \mathrm{OPT}_j(R).
\]
Although the individual quantities $\mathrm{OPT}_j(R)$ may vary with the choice of tie-breaking rule, the equality above implies that
\[
\sum_{j=1}^m \mathbb{E}_R[\mathrm{OPT}_j(R)]
\;=\;
\mathbb{E}_R[\mathrm{OPT}(R)],
\]
and thus the aggregate contribution across knapsacks is invariant.
The relative contribution of each knapsack can still vary substantially across different realizations and approximate solutions, which makes estimating the individual knapsack contributions $\E_R[\mathrm{OPT}_j]$ from global information alone highly challenging.

To address this issue, we assume oracle access to the expected knapsack-level optima $\E_R[\mathrm{OPT}_j]$ for all $j \in [m]$. Under this assumption, we establish the following performance guarantee:

\begin{theorem}[GAP Sparsifier]
\label{thm:main-approximation}
For parameters $\epsilon \in (0,1/6)$ and $p \in (0,1]$, assume oracle access to the expected knapsack optima $\E_R[\mathrm{OPT}_j]$ for each knapsack $j \in [m]$. Then Algorithm~\ref{alg:query_set_construction} produces a $\left(1 - 6\epsilon \right)$-approximate sparsifier for the stochastic GAP problem with sparsification degree
\[
O\left(\frac{\log^{2}(1/\epsilon)}{\epsilon^{3} p}\right).
\]
\end{theorem}

Since the Multiple Knapsack Problem is a special case of GAP, we obtain the following immediate corollary:

\begin{corollary}[Multiple Knapsack Sparsifier]
\label{cor:mkp-sparsification}
Under the same oracle assumption, Algorithm~\ref{alg:query_set_construction} produces a $\left(1 - 6\epsilon \right)$-approximate sparsifier for the stochastic Multiple Knapsack problem, with sparsification degree
\[
O\left(\frac{\log^{2}(1/\epsilon)}{\epsilon^{3} p}\right).
\]
\end{corollary}

\subsection{Relaxing Oracle Assumptions}

{In practice, obtaining precise offline computations of each $\E_R[\mathrm{OPT}_j]$ may be infeasible. We therefore analyze the robustness of our algorithm under weaker information settings.}

\paragraph{Approximate Oracle Access.} 
Given a $\beta$-approximation to the total stochastic optimum $\mathrm{OPT}$, we can estimate each $\E_R[\mathrm{OPT}_j]$ via expected marginal contributions of knapsacks in the approximate solution. Using these estimates, Algorithm~\ref{alg:query_set_construction} achieves a $\beta \cdot (1-6\epsilon)$-approximation with the same sparsification degree bound.
The analysis remains identical to Theorem~\ref{thm:main-approximation}, using the $\beta$-approximate assignment as the benchmark for the charging argument.

\paragraph{Global Oracle Access.}
A plausible scenario is having access to the global expectation $\E_R[\mathrm{OPT}]$ (or a constant factor estimate thereof) without granular knapsack-level details. In this case, we can uniformly distribute the expectation by setting $M_j = \E_R[\mathrm{OPT}]/m$ for all $j$. This forces the algorithm to cover a wider range of potential values per knapsack, introducing a logarithmic dependence on $m$ in the sparsification degree.

Intuitively, the contribution of any specific knapsack $j$ lies somewhere between a uniform share ($\approx \E_R[\mathrm{OPT}]/m$) and the total value ($\approx \E_R[\mathrm{OPT}]$). To ensure we capture the relevant items regardless of how the optimal solution distributes value, we set the minimum value threshold $M_j$ based on the uniform lower bound $\E_R[\mathrm{OPT}]/m$. Consequently, the bucket hierarchy for \emph{every} knapsack must span the expansive range from this uniform average up to the global maximum. This widens the value range by a factor of $m$, which, due to the geometric progression of bucket boundaries, incurs a logarithmic penalty in the number of buckets.

\begin{corollary}\label{cor:global-oracle-sparsifier}
Assume oracle access to the expected global optimum $\E_R[\mathrm{OPT}]$. Then the modified version of Algorithm~\ref{alg:query_set_construction} with $M_j = \E_R[\mathrm{OPT}]/m$ and 
\( K := \left\lceil \frac{2}{\epsilon^2} \log \left( \frac{m}{\epsilon^3} \right) \right\rceil \) is a $(1-6\epsilon)$-sparsifier with degree 
$O\left(\frac{\log^{2}(1/\epsilon)\log(m)}{\epsilon^{3} p}\right)$.
\end{corollary}


\subsection{Theoretical Boundaries and Practical Effectiveness}
The Exponential Time Hypothesis (ETH) \cite{eth} and sparsification algorithms pursue fundamentally opposing objectives. ETH 
implies that NP-hard problems cannot be solved substantially faster than exponential time, 
while sparsification aims to dramatically reduce problem size while preserving near-optimal solutions. This tension suggests that ETH establishes fundamental boundaries on sparsification effectiveness for computationally hard problems.

To investigate these boundaries, we analyzed the structure of APX-hard GAP instances established by Chekuri and Khanna \cite{chekuri2005mkp} through reductions from 3-dimensional matching \cite{3dmatching}. Our examination reveals that both hard instance families exhibit sparsification degree exactly 2 -- the entire item set already constitutes a sparse query set. This demonstrates that instances where GAP achieves maximum computational difficulty cannot benefit from sparsification, as they possess no redundancy to eliminate.

Unfortunately, our investigation of established GAP benchmark datasets \cite{beasley_dataset, dataset2} reveals that these instances are also sparse, with entire item sets remaining feasible and computational challenges arising from assignment optimization rather than item selection. Since both theoretically hard instances and existing benchmarks exhibit inherent sparsity, they provide limited insight into sparsification effectiveness.

We therefore designed controlled experiments with synthetic instances to systematically evaluate sparsification under varying redundancy levels. Our experimental design parameterized instances by item count $n \in \{1000, 2000, 5000, 10000\}$, knapsack count $m \in \{1,2,5\}$, and crucially, redundancy ratio $r$ (total items divided by optimal solution size). We sampled weights and values from uniform and normal distributions with correlation parameter $\rho$ to control their relationship.

We solve each instance using Gurobi~\cite{gurobi} as our Mixed Integer Linear Program (ILP) solver. Our evaluation compares three approaches: (A) solving the full instance with ILP, (B) applying sparsification and solving the sparsified instance with ILP, and (C) early-stopping the full ILP when the sparsification algorithm terminates.

\begin{figure}[!htb]
    \centering
    \fbox{%
        \begin{minipage}{\textwidth}
        \includegraphics[width=\linewidth]{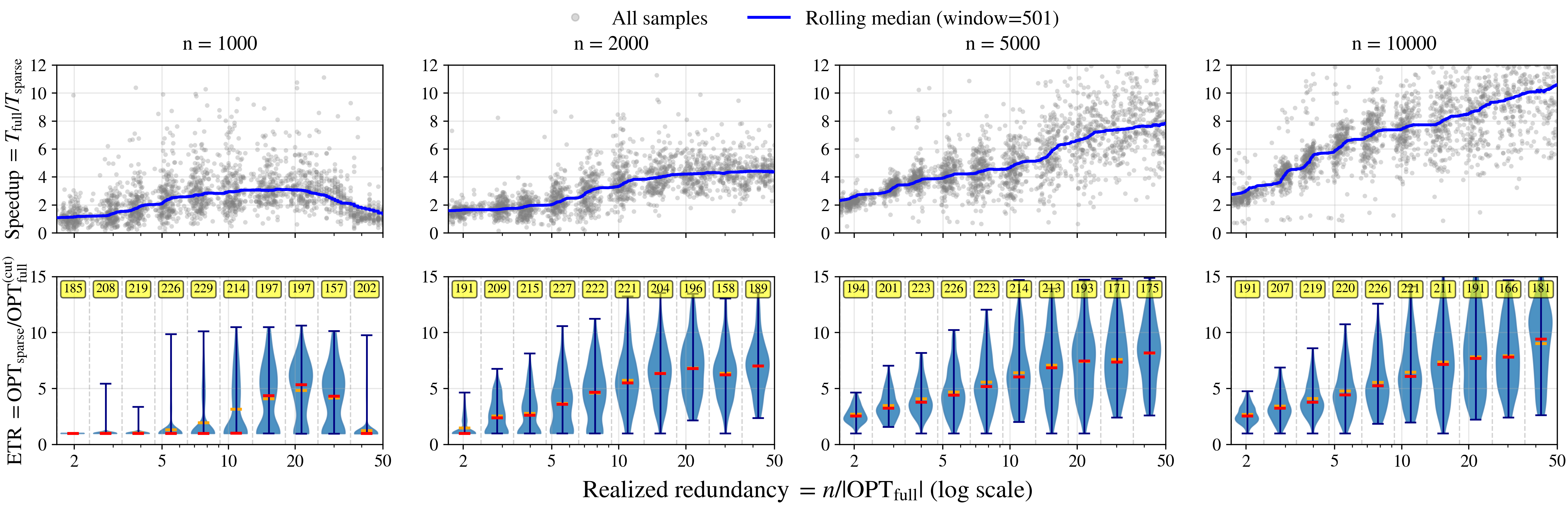}
        \caption{ 
            \label{fig:experiments_summary} 
            This figure presents experimental results for $n \in \{1000, 2000, 5000, 10000\}$. The upper panels illustrate speedup ratios (runtime of method A divided by runtime of method B) plotted against realized redundancy, defined as the ratio of total items to items in the optimal solution. Each gray point corresponds to an individual experiment, while the blue line represents the rolling median computed with a window size of $501$. The lower panels present violin plots showing the distribution of efficiency ratios (performance of method B divided by performance of method C) across varying levels of realized redundancy. Experiments are aggregated into discrete redundancy intervals (in $\log$-scale), with red dots marking the mean value and orange dots marking the median value within each bin. Sample sizes for each interval are indicated at the top of the corresponding violin plot, and the violin shapes visualize the complete distribution of efficiency ratios within each redundancy bin.
        }
    \end{minipage}
    }
\end{figure}

For instances with $m=2$ knapsacks and high redundancy ($r > 4$), sparsification achieves significant computational benefits \footnote{Interestingly, we observe counterintuitive scaling behavior in our experiments. While theory suggests that larger $m$ values should make GAP instances harder since they contain smaller $m$ instances as special cases, our synthetic data shows decreasing solution times as $m$ increases from $2$ to $5$. We attribute this phenomenon to our random instance generation methodology, where additional knapsacks appear to create easier optimization landscapes rather than harder ones. We provide detailed analysis of this behavior in Appendix~\ref{appendix:experiments}.}. Specifically, method B runs $4$ times faster than method A while maintaining $99\%$ of the solution quality. We also compare methods B and C to demonstrate sparsification's effectiveness within identical time budgets, finding that method B produces solutions that are $5$ times better than method C under the same conditions.

Figure~\ref{fig:experiments_summary} summarizes our key findings for the $m=2$ case, with complete experimental details provided in Appendix~\ref{appendix:experiments}.

\section{Analysis of GAP Sparsifier}

Before beginning the analysis, we briefly recall the relevant
notation. Non-bold symbols (e.g., $S$) denote sets of items, whereas bold symbols
(e.g., $\mathbf{S}$) denote assignment solutions, represented as sets of
item--knapsack pairs. For each realization $R$ of the active set, we previously
introduced a canonical optimal assignment $\mathbf{OPT}(R)$ via a deterministic
tie-breaking rule. This choice is used solely for analysis; the sparsification
algorithm (Algorithm~\ref{alg:query_set_construction}) never requires access to
$\mathbf{OPT}(R)$ for any individual $R$. For simplicity of notation, throughout this section we omit
the dependence on $R$ and write $\mathbf{OPT}$.

With this notation in place, our analysis centers on a reconstruction algorithm
that conceptually simulates $\mathbf{OPT}$ while constructing a feasible
assignment $\mathbf{ALG}$ using only queried active items. The reconstruction
processes buckets $(j,k)$ sequentially, maintaining a partial assignment
$\overline{\mathbf{OPT}}$ that incrementally grows toward $\mathbf{OPT}$.

\begin{algorithm}[h]
\caption{\textsc{ReconstructionProcedure}}
\label{alg:reconstruction}
\begin{algorithmic}[1]
\STATE $\overline{\mathbf{OPT}} \gets \emptyset$, $\mathbf{ALG} \gets \emptyset$
\FOR{$j = 1$ to $m$}
    \STATE $(\overline{\mathbf{OPT}}, \mathbf{ALG}) 
            \gets \textsc{FillSmallBucket}(\overline{\mathbf{OPT}}, \mathbf{ALG}, j, 0)$
    \FOR{$k = 1$ to $\left\lceil \frac{2}{\epsilon^{2}} \log \frac{1}{\epsilon^{3}} \right\rceil$}
        \STATE $(\overline{\mathbf{OPT}}, \mathbf{ALG}) 
            \gets \textsc{FillLargeBucket}(\overline{\mathbf{OPT}}, \mathbf{ALG}, j, k)$
    \ENDFOR
\ENDFOR
\STATE $(\overline{\mathbf{OPT}}, \mathbf{ALG}) \gets \textsc{FillAllSuperBuckets}(\overline{\mathbf{OPT}}, \mathbf{ALG})$
\end{algorithmic}
\end{algorithm}

The reconstruction algorithm employs three specialized subroutines. \textsc{FillLargeBucket}, for $1 \le k \le K$, replaces missed high-value items with lighter alternatives.
\textsc{FillSmallBucket}, for $k = 0$, uses density-based substitution for low-value items. \textsc{FillAllSuperBuckets}, for $k = K+1$, performs no substitutions and simply inserts each queried item using its original assignment while discarding all missed items.

Before providing definitions of these subroutines, we first state some desired properties of these subroutines and hence our \textsc{ReconstructionProcedure}. Based on these properties, we first prove the main result of this section in Section~\ref{sec:proof-of-main-thm}. Later, Section~\ref{sec:subroutines} will be devoted to formal definitions of subroutines,  \textsc{FillLargeBucket}, \textsc{FillSmallBucket}, and \textsc{FillAllSuperBuckets}, as well as proofs of their desired properties. 

To formally track the capacity usage and value gain during reconstruction, we define the total weight and value contributed to each knapsack by the subroutine calls in \textsc{ReconstructionProcedure}.
Consider an assignment \( \mathbf{S} \in \{ \mathbf{ALG}, \overline{\mathbf{OPT}} \} \).
For a single call to either \textsc{FillSmallBucket} or \textsc{FillLargeBucket}, we define:
\begin{itemize}
    \item \( \Delta w_j(\mathbf{S}) \): for each knapsack \( j \in [m] \), the increase in the total weight assigned to \( j \) in \( \mathbf{S} \), namely the sum of \( w_{ij} \) over all item--knapsack pairs \( (i,j) \) that are added to \( \mathbf{S} \) during this call;
    \item \( \Delta v(\mathbf{S}) \): the increase in the total value of \( \mathbf{S} \) across all knapsacks, namely the sum of \( v_{ij} \) over all pairs \( (i,j) \) that are added to \( \mathbf{S} \) during this call.
\end{itemize}



\paragraph{Feasibility.}
The \textsc{ReconstructionProcedure} maintains feasibility through three key properties:
(1) each call to the procedure adds only queried active items to $\mathbf{ALG}$, ensuring that $\mathbf{ALG} \subseteq Q \cap R$ at all times;
(2) no item is ever assigned more than once within either solution: once an item appears in $\mathbf{ALG}$ (respectively, $\overline{\mathbf{OPT}}$), it is never reassigned within that same solution;
and (3) for every knapsack $j$, the total weight assigned in $\mathbf{ALG}$ never exceeds that in $\overline{\mathbf{OPT}}$, ensuring that capacity constraints remain satisfied throughout.
Lemmas~\ref{lem:fill-large-less-capacity}, \ref{lem:fill-small-less-capacity}, and~\ref{lem:fill-super-less-capacity} formally establish these properties for large, small, and super buckets, respectively.

\begin{lemma}[Feasibility in Large Buckets]
\label{lem:fill-large-less-capacity}
When \textsc{FillLargeBucket} is called for knapsack $j$ and bucket~$k$,
the procedure assigns only queried active items to $\mathbf{ALG}$, and no item already
appearing in $\mathbf{ALG}$ (respectively, $\overline{\mathbf{OPT}}$)
is ever reassigned within that same solution.
Moreover, for every knapsack $j' \in [m]$,
\[
\Delta w_{j'}(\mathbf{ALG}) \le \Delta w_{j'}(\overline{\mathbf{OPT}}).
\]
\end{lemma}

\begin{lemma}[Feasibility in Small Buckets]
\label{lem:fill-small-less-capacity}
When \textsc{FillSmallBucket} is called for knapsack $j$ and bucket~$k$,
the procedure assigns only queried active items to $\mathbf{ALG}$, and no item already
appearing in $\mathbf{ALG}$ (respectively, $\overline{\mathbf{OPT}}$)
is ever reassigned within that same solution.
Moreover, for every knapsack $j' \in [m]$,
\[
\Delta w_{j'}(\mathbf{ALG}) \le \Delta w_{j'}(\overline{\mathbf{OPT}}).
\]
\end{lemma}

\begin{lemma}[Feasibility in Super Buckets]
\label{lem:fill-super-less-capacity}
When \textsc{FillAllSuperBuckets} is called, 
the procedure assigns only queried active items to $\mathbf{ALG}$, and no item already appearing
in $\mathbf{ALG}$ (respectively, $\overline{\mathbf{OPT}}$) is ever reassigned within that same solution. Moreover, for every knapsack $j' \in [m]$,
\[
\Delta w_{j'}(\mathbf{ALG}) \le \Delta w_{j'}(\overline{\mathbf{OPT}}).
\]
\end{lemma}



\paragraph{Approximation:} 
The \textsc{ReconstructionProcedure} adds new item--knapsack assignments to both $\overline{\mathbf{OPT}}$ and $\mathbf{ALG}$ in such a way that the total value 
of the new assignments added to $\mathbf{ALG}$ closely tracks that of the new assignments added to $\overline{\mathbf{OPT}}$. This ensures that, throughout the reconstruction process, the value of $\mathbf{ALG}$ remains a good approximation to the value of $\overline{\mathbf{OPT}}$. The following three lemmas prove this property for large-value, small-value, and super-value buckets, respectively.

\begin{lemma}[Approximation Guarantee in Large Buckets]
\label{lem:fill-large-value}
When \textsc{FillLargeBucket} is called for knapsack $j$ and bucket $k$, the following inequality holds:
\[
\mathbb{E}_R[\Delta v(\mathbf{ALG})] 
\ge
(1-2\epsilon)\,\Delta v(\overline{\mathbf{OPT}}).
\]
\end{lemma}

\begin{lemma}[Approximation Guarantee in Small Buckets]
\label{lem:fill-small-value}
When \textsc{FillSmallBucket} is called for knapsack $j$ and bucket $k$, the inequality
\[
\mathbb{E}_R[\Delta v(\mathbf{ALG})]
\ge
(1-2\epsilon)\,\Delta v(\overline{\mathbf{OPT}}) - \epsilon^2 M_j
\]
holds, where $M_j = \mathbb{E}_R[\mathrm{OPT}_j]$ denotes the expected contribution of knapsack~$j$ to the optimal assignment.
\end{lemma}

\begin{lemma}[Approximation Guarantee in Super Buckets]
\label{lem:fill-super-value}
Assume $0 < \epsilon \le 1/2$.
When \textsc{FillAllSuperBuckets} is called, the following inequality holds:
\[
\mathbb{E}_R\left[\Delta v(\overline{\mathbf{OPT}})\right]
-
\mathbb{E}_R\left[\Delta v(\mathbf{ALG})\right]
\le
3\epsilon \cdot \mathbb{E}_R\left[v(\mathbf{OPT})\right].
\]
\end{lemma}

\paragraph{Correct Benchmark:} 
Finally, to demonstrate that the final output $\mathbf{ALG}$ is a good 
approximation of $\mathbf{OPT}$, we ensure that 
\textsc{ReconstructionProcedure} terminates with 
$\overline{\mathbf{OPT}}=\mathbf{OPT}$. 
The following lemma establishes this property:

\begin{lemma}
\label{lem:baropt-equals-opt}
Upon completion of \textsc{ReconstructionProcedure} 
(Algorithm~\ref{alg:reconstruction}), we have 
$\overline{\mathbf{OPT}} = \mathbf{OPT}$.
\end{lemma}

Now, we are ready to prove the main result of this section using these properties. Formal definitions of subroutines \textsc{FillSmallBucket}, \textsc{FillLargeBucket}, and \textsc{FillAllSuperBuckets}, as well as proofs of Lemmas~\ref{lem:fill-large-less-capacity}-\ref{lem:baropt-equals-opt} will be provided in Section~\ref{sec:subroutines}.

\subsection{Proof of Theorem~\ref{thm:main-approximation}}
\label{sec:proof-of-main-thm}
We now prove the main result of this section using the lemmas stated above.

\begin{proof}[Proof of Theorem~\ref{thm:main-approximation}]

Fix any realization \( R \subseteq E \).


We first bound the size of the query set by analyzing the total weight selected for each knapsack. Algorithm~\ref{alg:query_set_construction} maintains 
$K = \left\lceil \frac{2}{\epsilon^{2}} \log\left(\frac{1}{\epsilon^{3}}\right) \right\rceil = O\left(\frac{1}{\epsilon^{2}}\log\frac{1}{\epsilon}\right)$
buckets per knapsack and iterates for $\alpha = \lceil 1/\epsilon \rceil$ rounds. In each specific round $t$ and bucket $(j,k)$, the algorithm selects items with a total weight of at most $(\frac{\tau(\epsilon^2)}{p} + 1) C_j$, where the $+1$ accounts for the discrete weight of the final item. Aggregating over all rounds and buckets, the total weight implicitly associated with knapsack $j$ is bounded by:
\[
w_j(Q) \le \alpha \cdot K \cdot O\left(\frac{\tau(\epsilon^2)}{p}\right) \cdot C_j = 
{
O\left(\frac{\log^{2}(1/\epsilon)}{\epsilon^{3} p}\right) C_j
}
\]
This weight bound establishes that the indicator vector $\mathbf{1}_Q$ lies within the natural linear programming relaxation of the problem, scaled by a factor $d_{LP} = 
{
O\left(\frac{\log^{2}(1/\epsilon)}{\epsilon^{3} p}\right)
}$. To translate this into the required polyhedral sparsification degree (which is defined with respect to $\mathcal{P}_{\mathcal{F}}$, the convex hull of \emph{integer} feasible solutions), we leverage the fact that the integrality gap of the standard linear relaxation for GAP is at most 2 (assuming feasible singletons). This implies polytope of relaxed LP $\mathcal{P}_{LP}$ is contained by $2 \cdot \mathcal{P}_{\mathcal{F}}$. Consequently, the sparsification degree is at most $2 d_{LP}$, which preserves the asymptotic bound:
\[
d = 
O\left(\frac{\log^{2}(1/\epsilon)}{\epsilon^{3} p}\right).
\]

\paragraph{Feasibility.}
For each procedure call to \textsc{FillLargeBucket}, \textsc{FillSmallBucket}, or \textsc{FillAllSuperBuckets}, let $\Delta w^{j,k}_{j'}(\cdot)$ denote the weight-change function for knapsack $j'$. Applying Lemmas~\ref{lem:fill-large-less-capacity}--\ref{lem:fill-super-less-capacity}  and summing over all $(j,k)$, we obtain
\[
w_{j'}(\mathbf{ALG})
    = \sum_{j,k} \Delta w^{j,k}_{j'}(\mathbf{ALG})
    \;\le\;
      \sum_{j,k} \Delta w^{j,k}_{j'}(\overline{\mathbf{OPT}})
    = w_{j'}(\mathbf{OPT}),
\]
where the final equality follows from Lemma~\ref{lem:baropt-equals-opt}.

Moreover, Lemma~\ref{lem:fill-large-less-capacity} and 
Lemma~\ref{lem:fill-small-less-capacity} ensure that 
$\mathbf{ALG}$ includes each item only once and that every item inserted into $\mathbf{ALG}$ comes from the queried active set $Q \cap R$, guaranteeing that $\mathbf{ALG}$ is a valid solution after sparsification and realization.
Since $\mathbf{OPT}$ is feasible, the weight domination
$w_{j'}(\mathbf{ALG}) \le w_{j'}(\mathbf{OPT})$ implies that
$\mathbf{ALG}$ is also feasible.
\paragraph{Approximation Guarantee.}
For each knapsack $j$ and bucket index $k \le K$, let $\Delta v^{j,k}$ denote 
the value increase function when calling \textsc{FillLargeBucket} or 
\textsc{FillSmallBucket}. For the super bucket $k=K+1$, we define 
$\Delta v^{j,K+1}$ analogously as the value contribution of 
\textsc{FillAllSuperBuckets} for knapsack $j$.
By Lemmas~\ref{lem:fill-large-value}--\ref{lem:baropt-equals-opt}, we have
\begin{align}
\E_R[v(\mathbf{ALG})]
&= \sum_{j}\sum_{k=0}^{K} \E_R[\Delta v^{j,k}(\mathbf{ALG})] + \sum_{j} \E_R[\Delta v^{j,K+1}(\mathbf{ALG})] \label{eq:linearity} \\
&\ge \Bigl(
      (1-2\epsilon)\sum_{j}\sum_{k=0}^{K} \E_R[\Delta v^{j,k}(\overline{\mathbf{OPT}})]
      \;-\; \epsilon^2 \sum_j M_j
    \Bigr) \notag\\
& ~~~~\;+\;
    \Bigl(
      \sum_{j}\E_R[\Delta v^{j,K+1}(\overline{\mathbf{OPT}})]
      \;-\; 3\epsilon\,\E_R[v(\mathbf{OPT})]
    \Bigr)
    \label{eq:value-lemma-bounds}\\
&\ge (1-2\epsilon)\sum_{j}\sum_{k=0}^{K+1} \E_R[\Delta v^{j,k}(\overline{\mathbf{OPT}})]
      \;-\; \epsilon^2 \sum_j M_j
      \;-\; 3\epsilon\,\E_R[v(\mathbf{OPT})] \label{eq:rearrange}\\
&= (1-2\epsilon)\E_R[v(\mathbf{OPT})]
    \;-\; \epsilon^2 \E_R[v(\mathbf{OPT})]
    \;-\; 3\epsilon\,\E_R[v(\mathbf{OPT})]
    \label{eq:lemma-benchmark} \\
&= (1 - 2\epsilon - \epsilon^2 - 3\epsilon)\E_R[v(\mathbf{OPT})]
    \nonumber \\
&\ge (1-6\epsilon)\E_R[v(\mathbf{OPT})], \label{eq:final-bound-super}
\end{align}
where equation~\eqref{eq:linearity} follows from linearity of expectation; equation~\eqref{eq:value-lemma-bounds} applies Lemmas~\ref{lem:fill-large-value} and~\ref{lem:fill-small-value} to derive the first parenthesized term, and applies Lemma~\ref{lem:fill-super-value} to derive the second parenthesized term; equation~\eqref{eq:rearrange} rearranges the terms; and equation~\eqref{eq:lemma-benchmark} uses $\E_R[v(\mathbf{OPT})] = \sum_j M_j$ together with Lemma~\ref{lem:baropt-equals-opt}. The final inequality holds for all $0 < \epsilon \le 1/6$.


\end{proof}

\section{Reconstruction Procedures}
\label{sec:subroutines}

This section presents the two reconstruction subroutines comprising \textsc{ReconstructionProcedure}, which incrementally construct both the optimal solution $\overline{\mathbf{OPT}}$ and a feasible algorithmic solution $\mathbf{ALG}$ based on queried active elements. We begin by establishing the notation required for our theoretical guarantees.

\subsection{Notation and Preliminaries}
\label{sec:notation}
Our analysis systematically partitions the optimal solution and tracks queried elements across value regimes.

\paragraph{Basic Definitions.}
\begin{itemize}
    \item $v_i^{\OPT}$: The value of item $i$ under the optimal assignment:
    \[
    v_i^{\OPT} = 
    \begin{cases}
        v_{ij} & \text{if there exists } j \text{ such that } (i, j) \in \mathbf{OPT}, \\
        0      & \text{otherwise}.
    \end{cases}
    \]
\end{itemize}

\paragraph{Buckets and Query Sets.}
\begin{itemize}

    \item $B_{j,k}$: 
    The collection of \emph{item--knapsack pairs} $(i,j)$ such that, when item $i$
    is assigned to knapsack $j$, its value $v_{ij}$ falls into the $k$-th value
    scale.

    \item $\bar{B}_{j,k,t}$:
    The set of active \emph{items} $i$ that are queried by 
    Algorithm~\ref{alg:query_set_construction} in round~$t$ via bucket $(j,k)$.
    Formally, $i \in \bar{B}_{j,k,t}$ iff there exists a pair $(i,j)\in B_{j,k}$
    selected in round~$t$ (i.e., the algorithm selects $(i,j)$ with
    $\beta(i,j)=k$ and updates $Q \leftarrow Q \cup \{i\}$).
    By construction, the item sets $\bar{B}_{j,k,t}$ over all $(j,k,t)$ are
    pairwise disjoint.

    \item $\bar{B}_{j,k} = \bigcup_{t=1}^{\alpha-1} \bar{B}_{j,k,t}$:
    The set of all active \emph{items} queried via bucket $(j,k)$ across all rounds.
    The disjointness over $t$ implies that the item sets $\bar{B}_{j,k}$ are also
    disjoint across different $(j,k)$.

\end{itemize}

\paragraph{Partition of Optimal Solution.}
We partition the items used by the optimal assignment $\mathbf{OPT}$ according to
whether they are captured by the query set.  Throughout, an item $i$ is said to
belong to $\OPT$ if there exists a knapsack $j' \in [m]$ such that
$(i,j') \in \mathbf{OPT}$.

\begin{itemize}
    \item 
    $S_{j,k}^{\mathrm{queried}}
    :=
    \{\, i \mid \exists j' \text{ with } (i,j') \in \mathbf{OPT},\; i \in \OPT \subseteq R,\; i \in \bar{B}_{j,k} \subseteq Q \cap R \,\}$:
    the set of items that appear in $\OPT$ and are queried via bucket $(j,k)$.  
    Note that an item $i$ may be assigned to some knapsack $j'$ in $\mathbf{OPT}$, yet still be queried through bucket $(j,k)$; the querying knapsack $j$ need not equal the assignment knapsack $j'$. Note that $S_{j,k}^{\mathrm{queried}} \subseteq \bar{B}_{j,k} \subseteq Q \cap R$.

    \item 
    $S_{j,k}^{\mathrm{missed}}
    :=
    \{\, i \mid (i,j) \in \mathbf{OPT},\; i \in \OPT \subseteq R,\; (i,j) \in B_{j,k},\; i \notin Q \,\}$:
    the set of items assigned to knapsack $j$ in $\mathbf{OPT}$ whose pair $(i,j)$ lies in bucket $B_{j,k}$, but $i$ is not queried in any round.

\end{itemize}

The union of two collections $\{S_{j,k}^{\mathrm{queried}}\}_{j,k}$ and $\{S_{j,k}^{\mathrm{missed}}\}_{j,k}$ forms a partition of the optimal solution $\OPT$:
\[
\OPT = \bigcup_{j,k} \left( S_{j,k}^{\mathrm{queried}} \cup S_{j,k}^{\mathrm{missed}} \right)
\]
where the union is disjoint.
To justify completeness, we rely on the fact that for every knapsack $j$, the bucket family $\{B_{j,k}\}_{k=0}^{K+1}$ covers the entire nonnegative value range. Bucket $0$ begins at value $0$, buckets $1$ through $K$ cover the geometric intervals up to $M_j/\epsilon$, and the super bucket $K\!+\!1$ collects all remaining items with $v_{ij} > M_j/\epsilon$.
Hence every item--knapsack pair chosen by $\mathbf{OPT}$ necessarily lies in some 
bucket $B_{j,k}$.  
Thus every item in $\mathbf{OPT}$ is covered either by 
$S_{j,k}^{\mathrm{queried}}$ or by $S_{j,k}^{\mathrm{missed}}$, establishing the 
claimed partition.

\paragraph{Excess Weight Event.}
\begin{itemize}
   \item $\mathcal{E}_{j,k,t} := \{ w_j(\bar{B}_{j,k,t}) \geq C_j \}$: Excess weight event for bucket $(j,k,t)$. Recall that $w_j(\bar{B}_{j,k,t}) = \sum_{i \in \bar{B}_{j,k,t}} w_{ij}$. Notice that, by construction, whenever $S_{j,k}^{\mathrm{missed}} \ne \emptyset$, bucket $(j,k)$ must exhaust its total query capacity $\tau(\epsilon^2)\, C_j$ in every round $t$. Therefore, by Lemma~\ref{lem:bucket_weight_guarantee}, we have $\Pr[\mathcal{E}_{j,k,t}] \geq 1 - \epsilon^2$ whenever $S_{j,k}^{\mathrm{missed}} \ne \emptyset$. 
\end{itemize}

\subsection{Subroutine Design}
Both subroutines handle different value regimes using an essentially same approach. For each bucket $(j,k)$:

\begin{itemize}

    \item \textbf{No missed items} ($S_{j,k}^{\mathrm{missed}} = \emptyset$):  
    Assign all queried items $S_{j,k}^{\mathrm{queried}}$ to both $\overline{\mathbf{OPT}}$ and $\mathbf{ALG}$ using the same item--knapsack assignments as in $\mathbf{OPT}$. No substitute items are needed in this case.

    \item \textbf{Missed items exist} ($S_{j,k}^{\mathrm{missed}} \neq \emptyset$):  
    Assign both $S_{j,k}^{\mathrm{queried}}$ and $S_{j,k}^{\mathrm{missed}}$ to $\overline{\mathbf{OPT}}$, using exactly the item--knapsack assignments they have in $\mathbf{OPT}$.  
    For $\mathbf{ALG}$, some items in $S_{j,k}^{\mathrm{missed}}$ are substituted by suitable queried alternatives from $\bar{B}_{j,k}$ according to the replacement rule of the procedure. All items in $S_{j,k}^{\mathrm{queried}}$ that are not used as replacements are then assigned to $\mathbf{ALG}$ using the same item--knapsack assignments as in $\mathbf{OPT}$.
\end{itemize}

The substitution strategies differ between regimes:

\paragraph{High-Value Regime ($1 \le k \le K$).} 
\textsc{FillLargeBucket} searches for lighter queried items to substitute each missed item $i \in S_{j,k}^{\mathrm{missed}}$:

\begin{enumerate}
    \item \textbf{Direct substitution:}  
    If there exists a lighter queried item that is not used by $\mathbf{OPT}$, then the reconstruction procedure assigns it to knapsack $j$ in place of $i$ in $\mathbf{ALG}$, while assigning $i$ to $\overline{\mathbf{OPT}}$ with its original assignment.
    
    \item \textbf{Redistribution:}  
    If no such unused substitute exists, the algorithm selects a set of lighter queried items that are used in $\mathbf{OPT}$, each of which can be reassigned to knapsack $j$ to substitute for $i$.  
    The algorithm forms a bundle consisting of $i$ together with these selected items; with high probability this bundle has size $1/\epsilon$.  
    The entire bundle is assigned to $\overline{\mathbf{OPT}}$ with their original assignment.  
    For $\mathbf{ALG}$, the algorithm compares two options:  
    (i) keeping all queried items in the bundle with their original $\mathbf{OPT}$ assignments, or  
    (ii) reassigning exactly one item to knapsack $j$ as the substitute for $i$.
    It chooses the option that yields the higher total value.
    This guarantees that $\mathbf{ALG}$ sacrifices at most the value of the least valuable item in the bundle, and therefore incurs at most a $1 - \frac{1}{\epsilon}$ fractional loss.
\end{enumerate}
Finally, the algorithm assigns all remaining queried items that appear in $\mathbf{OPT}$ to both $\overline{\mathbf{OPT}}$ and $\mathbf{ALG}$ using their original item--knapsack assignments.

\begin{algorithm}[!htb]
\caption{\textsc{FillLargeBucket}$(\overline{\mathbf{OPT}}, \mathbf{ALG}, j, k)$}
\label{alg:fill-large}
\begin{algorithmic}[1]
\WHILE{$S_{j,k}^{\mathrm{missed}} \neq \emptyset$} \label{line:large-while}
    \STATE Pick arbitrary $i \in S_{j,k}^{\mathrm{missed}}$
    \IF{there exists $i' \in \bar{B}_{j,k}$ such that $i' \notin \mathbf{OPT}$ and $i' \notin \mathbf{ALG}$} \label{line:large-if-available}
        \STATE $\overline{\mathbf{OPT}} \gets \overline{\mathbf{OPT}} \cup \{(i, j)\}$ and $\mathbf{ALG} \gets \mathbf{ALG} \cup \{(i', j)\}$
        \STATE $S_{j,k}^{\mathrm{missed}} \gets S_{j,k}^{\mathrm{missed}} \setminus \{i\}$ \label{line:large-case0-done} 
    \ELSE
        \FOR{$t = 1$ to $\alpha - 1$
        \label{line:large-t-loop}}
            \STATE \textbf{If} there exists $i'_t \in \bar{B}_{j,k,t}$ such that $(i'_t, j'_t) \in \mathbf{OPT}$ and $i'_t \notin \mathbf{ALG}$, \textbf{then} select such $i'_t$ and include index $t$ into set $T$.
        \ENDFOR
        
        \IF{$T \neq \emptyset$} \label{line:large-if-nonempty-T}
            \STATE $t^* \gets \mathop{\arg\min}\limits_{t \in T} \left( v^{\OPT}_{i'_t} \right)$ \label{line:large-argmin}
            \IF{$v^{\OPT}_{i'_{t^*}} \geq v^{\OPT}_i$} \label{line:large-case-a}
                \STATE $\mathbf{ALG} \gets \mathbf{ALG} \cup \left( \bigcup_{t \in T} (i'_t, j'_t) \right)$ \label{line:large-case-a-done}
            \ELSE \label{line:large-case-b}
                \STATE $\mathbf{ALG} \gets \mathbf{ALG} \cup \left( \bigcup_{t \in T \setminus \{t^*\}} (i'_t, j'_t) \right) \cup \{(i'_{t^*}, j)\}$ \label{line:large-case-b-done}
            \ENDIF
        \ENDIF
        \STATE $\overline{\mathbf{OPT}} \gets \overline{\mathbf{OPT}} \cup \left( \bigcup_{t \in T} (i'_t, j'_t) \right) \cup \{(i, j)\}$ \label{line:large-update-opt}
        \STATE $S_{j,k}^{\mathrm{queried}} \gets S_{j,k}^{\mathrm{queried}} \setminus \left\{ i'_t \mid t \in T \right\}$ \label{line:large-update-queried}
        \STATE $S_{j,k}^{\mathrm{missed}} \gets S_{j,k}^{\mathrm{missed}} \setminus \{i\}$ \label{line:large-update-missed}
    \ENDIF
\ENDWHILE
\FORALL{$i \in S_{j,k}^{\mathrm{queried}}$} \label{line:large-final-loop}
    \STATE Find $(i, j') \in \mathbf{OPT}$
    \STATE $\overline{\mathbf{OPT}} \gets \overline{\mathbf{OPT}} \cup \{(i, j')\}$ and $\mathbf{ALG} \gets \mathbf{ALG} \cup \{(i, j')\}$
    \STATE $S_{j,k}^{\mathrm{queried}} \gets S_{j,k}^{\mathrm{queried}} \setminus \{i\}$ \label{line:large-final-loop-done}
\ENDFOR
\RETURN $(\overline{\mathbf{OPT}}, \mathbf{ALG})$ \label{line:large-return}
\end{algorithmic}
\end{algorithm}

\paragraph{Low-Value Regime ($k = 0$).} 
\textsc{FillSmallBucket} uses density-based substitution, exploiting superior value-to-weight ratios:
\begin{itemize}
\item \textbf{No missed items:} All items in \( S_{j,0}^{\mathrm{queried}} \) are assigned to both solutions with their original assignments.

\item \textbf{Missed items exist:} The algorithm sorts the queried items in \( \bar{B}_{j,0} \) by their value contributions under \( \mathbf{OPT} \) divided by their weights in knapsack \( j \), i.e., \(\frac{v_i^{\OPT}}{w_{ij}}\), and constructs a prefix set \( S \) satisfying the following:
\begin{enumerate}
    \item The total weight of \( S \) does not exceed that of the missed items;
    \item Each item in \( S \) has a strictly larger ratio \( \frac{v_{ij}}{w_{ij}} \) than every item in \( S_{j,0}^{\mathrm{missed}} \);
    \item Each item in \( S \) provides more value when assigned to knapsack \( j \) than it does in \(\OPT\).
\end{enumerate}
Items in \(S_{j,0}^{\mathrm{queried}} \) and \(S_{j,0}^{\mathrm{missed}} \) are then assigned to \( \overline{\mathbf{OPT}} \) with their original assignments. For \( \mathbf{ALG} \), queried items not in \( S \) retain their original assignments, while items in \( S \) are reassigned to knapsack \( j \) to serve as substitutes for the missed items. This guarantees that \( \mathbf{ALG} \) sacrifices essentially only the value of the least valuable items among those assigned to \( \overline{\mathbf{OPT}} \). 
\end{itemize}

\begin{algorithm}[!htb]
\caption{\textsc{FillSmallBucket}$(\overline{\mathbf{OPT}}, \mathbf{ALG}, j, k)$}
\label{alg:fill-small}
\begin{algorithmic}[1]
\IF{$S_{j,k}^{\mathrm{missed}} = \emptyset$} \label{line:small-if-easy}
    \FORALL{$i \in S_{j,k}^{\mathrm{queried}}$} \label{line:small-loop-easy}
        \STATE Find $(i, j') \in \mathbf{OPT}$
        \STATE $\overline{\mathbf{OPT}} \gets \overline{\mathbf{OPT}} \cup \{(i, j')\}$ and $\mathbf{ALG} \gets \mathbf{ALG} \cup \{(i, j')\}$
        \STATE $S_{j,k}^{\mathrm{queried}} \gets S_{j,k}^{\mathrm{queried}} \setminus \{i\}$
    \ENDFOR
\ELSE \label{line:small-if-not-easy}
    \STATE Sort all $i \in \bar{B}_{j,k}$ by non-decreasing $\frac{v_i^{\OPT}}{w_{ij}}$ \label{line:small-sort}
    \STATE Let $S \subseteq \bar{B}_{j,k}$ be the largest prefix such that: \label{line:small-define-s} 
    \STATE \quad (i) $v_{ij} > v_i^{\OPT}$ for all $i \in S$
    \STATE \quad (ii) $w_j(S) := \sum_{i \in S} w_{ij} \leq \sum_{i \in S_{j,k}^{\mathrm{missed}}} w_{ij}$
    \FORALL{$i \in S_{j,k}^{\mathrm{queried}}$} \label{line:small-loop-main}
        \STATE Find $(i, j') \in \mathbf{OPT}$
        \STATE $\overline{\mathbf{OPT}} \gets \overline{\mathbf{OPT}} \cup \{(i, j')\}$
        \STATE $S_{j,k}^{\mathrm{queried}} \gets S_{j,k}^{\mathrm{queried}} \setminus \{i\}$
        \IF{$i \notin S$} \label{line:small-if-not-in-s}
            \STATE $\mathbf{ALG} \gets \mathbf{ALG} \cup \{(i, j')\}$
        \ELSE \label{line:small-else-in-s}
            \STATE $\mathbf{ALG} \gets \mathbf{ALG} \cup \{(i, j)\}$
            \STATE $S \gets S \setminus \{i\}$
        \ENDIF
    \ENDFOR
    \FORALL{$i \in S$} \label{line:small-loop-s}
        \STATE $\mathbf{ALG} \gets \mathbf{ALG} \cup \{(i, j)\}$
        \STATE $S \gets S \setminus \{i\}$
    \ENDFOR
    \FORALL{$i \in S_{j,k}^{\mathrm{missed}}$} \label{line:small-loop-s2}
        \STATE $\overline{\mathbf{OPT}} \gets \overline{\mathbf{OPT}} \cup \{(i, j)\}$
        \STATE $S_{j,k}^{\mathrm{missed}} \gets S_{j,k}^{\mathrm{missed}} \setminus \{i\}$
    \ENDFOR
\ENDIF
\RETURN $(\overline{\mathbf{OPT}}, \mathbf{ALG})$ \label{line:small-return}
\end{algorithmic}
\end{algorithm}

\paragraph{Super-Bucket Regime ($k = K\!+\!1$).}
In the super bucket, \textsc{FillAllSuperBuckets} follows the simplest rule among all subroutines and performs no substitutions.

\begin{enumerate}
    \item \textbf{Queried super items.}
    For every $i \in S^{\mathrm{queried}}_{j,K+1}$, the algorithm assigns it to both $\overline{\mathbf{OPT}}$ and $\mathbf{ALG}$ using its original item--knapsack assignment.

    \item \textbf{Missed super items.}
    For each missed item $i \in S^{\mathrm{missed}}_{j,K+1}$, the algorithm adds 
    $(i,j)$ only to $\overline{\mathbf{OPT}}$.  
    The algorithm leaves $\mathbf{ALG}$ unchanged and ignores all such items.
\end{enumerate}

Although this rule may appear wasteful, since every missed super item has relatively large value in its assigned bucket, the total loss incurred in the super-bucket regime will be shown to be globally very limited.

\begin{algorithm}[!htb]
\caption{\textsc{FillAllSuperBuckets}$(\overline{\mathbf{OPT}}, \mathbf{ALG})$}
\label{alg:fill-super}
\begin{algorithmic}[1]
\FOR{$j = 1$ to $m$}
    \FORALL{$i \in S^{\mathrm{queried}}_{j,K+1}$}
        \STATE Find $(i, j') \in \mathbf{OPT}$
        \STATE $\overline{\mathbf{OPT}} \gets \overline{\mathbf{OPT}} \cup \{(i, j')\}$
        \STATE $\mathbf{ALG} \gets \mathbf{ALG} \cup \{(i, j')\}$
        \STATE $S^{\mathrm{queried}}_{j,K+1} \gets S^{\mathrm{queried}}_{j,K+1} \setminus \{i\}$
    \ENDFOR

    \FORALL{$i \in S^{\mathrm{missed}}_{j,K+1}$}
        \STATE $\overline{\mathbf{OPT}} \gets \overline{\mathbf{OPT}} \cup \{(i, j)\}$
        \STATE $S^{\mathrm{missed}}_{j,K+1} \gets S^{\mathrm{missed}}_{j,K+1} \setminus \{i\}$
    \ENDFOR

\ENDFOR

\RETURN $(\overline{\mathbf{OPT}}, \mathbf{ALG})$
\end{algorithmic}
\end{algorithm}







With both subroutines formally defined, we now prove the key properties of the reconstruction procedure, conditioned on the event that each queried bucket contains sufficient active weight (as ensured with high probability by Lemma~\ref{lem:bucket_weight_guarantee} when $S_{j,k}^{\mathrm{missed}} \neq \emptyset$). 

\subsection{Substitution Guarantee for Missed Items}
\begin{lemma}[Properties of Missed Items]
\label{lem:missed-properties}
Condition on the event $S_{j,k}^{\mathrm{missed}} \neq \emptyset$, if $\mathcal{E}_{j,k,t}$ holds then the following properties hold for all $t$:
\begin{enumerate}
    \item \textsc{BucketFilled}: The total weight of each queried bucket satisfies
    \[
    w_j(\bar{B}_{j,k,t}) \geq C_j.
    \]
    
    \item \textsc{AlternativeExists}:
    \begin{itemize}
        \item If $k \neq 0$, then for all $i' \in \bar{B}_{j,k}$ and all $i \in S_{j,k}^{\mathrm{missed}}$,
        \[
        w_{i'j} \leq w_{ij}.
        \]
        \item If $k = 0$, then for all $i' \in \bar{B}_{j,k}$ and all $i \in S_{j,k}^{\mathrm{missed}}$,
        \[
        \frac{v_{i'j}}{w_{i'j}} \geq \frac{v_{ij}}{w_{ij}}.
        \]
    \end{itemize}
    
    \item \textsc{SizeMatch} (for $k \neq 0$): The queried bucket contains at least as many items as missed:
    \[
    \left| \bar{B}_{j,k,t} \right| \geq \left| S_{j,k}^{\mathrm{missed}} \right|.
    \]
\end{enumerate}
\end{lemma}

\begin{proof}
\begin{enumerate}
    \item Holds by the definition of event $\mathcal{E}_{j,k,t}$, which includes \( w_j(\bar{B}_{j,k,t}) \geq C_j \).
    
    \item The algorithm selects items from bucket $B_{j,k}$ in increasing weight order (for $k \neq 0$) or decreasing value-density order (for $k = 0$). Since missed items were not selected, all chosen items dominate them in the respective ordering.
    
    \item For $k \neq 0$, since queried items are lighter than missed items and $\bar{B}_{j,k,t}$ achieves total weight $\geq C_j$, the cardinality of $\bar{B}_{j,k,t}$ must exceed that of $S_{j,k}^{\mathrm{missed}}$.
\end{enumerate}
\end{proof}

Lemma~\ref{lem:missed-properties} ensures successful substitution for missed items in the query set. The property \textsc{SizeMatch} ensures sufficient substitutes for high-value buckets while the property \textsc{BucketFilled} provides adequate total weight for low-value bucket substitutions. Finally, the property \textsc{AlternativeExists} guarantees each missed item has a superior substitute (lighter weight or higher density).

\subsection{Feasibility Analysis} 

We prove that both subroutines maintain feasibility by ensuring that $\mathbf{ALG}$ exclusively assigns queried items and and that its capacity consumption is dominated by $\overline{\mathbf{OPT}}$ in every knapsack. Furthermore, we confirm that both solutions strictly adhere to the matching constraint, assigning each item at most once.

\begin{proof}[Proof of Lemma~\ref{lem:fill-large-less-capacity}]
Consider a specific invocation of \textsc{FillLargeBucket} for bucket $(j,k)$ and recall that $S_{j,k}^{\mathrm{queried}} \subseteq \bar{B}_{j,k} \subseteq Q \cap R$.
By construction, both $\overline{\mathbf{OPT}}$ and $\mathbf{ALG}$ assigns items exclusively from $\bar{B}_{j,k}$, and no item is assigned multiple times during this call.  Since the aggregate bucket sets $\{\bar{B}_{j,k}\}_{j,k}$ are pairwise disjoint across all $(j,k)$ (Section~\ref{sec:notation}), we guarantee that no item is assigned more than once across the entire reconstruction process. In addition, every item added to $\mathbf{ALG}$ comes from the queried active set $Q \cap R$.

We now proceed to bound the capacity consumption. We analyze the weight increments $\Delta w_{j'}(\cdot)$ -- representing the total weight added to knapsack $j'$ during this specific call -- by examining each case in Algorithm~\ref{alg:fill-large} separately.

\paragraph{Direct Substitution: (Lines~\ref{line:large-if-available} to~\ref{line:large-case0-done}).}
In this case, the algorithm adds $(i, j)$ to $\overline{\mathbf{OPT}}$ and $(i', j)$ to $\mathbf{ALG}$. Since $S_{j,k}^{\mathrm{missed}} \neq \emptyset$, property \textsc{AlternativeExists} guarantees $w_{i'j} \le w_{ij}$. No other knapsacks are affected, so
\[
\Delta w_{j'}(\mathbf{ALG}) \le \Delta w_{j'}(\bar{\mathbf{OPT}}) \quad \text{for all } j'.
\]

\paragraph{Value-Based Rejection: (Lines~\ref{line:large-case-a} to~\ref{line:large-case-a-done}).} 
Let \( T \subseteq [\alpha - 1] \) be the set of indices selected during the loop on Line~\ref{line:large-t-loop}, and for each \( t \in T \), let \( (i'_t, j'_t) \in \mathbf{OPT} \) denote the assignment recovered from bucket \( \bar{B}_{j,k,t} \). 

\noindent $\overline{\mathbf{OPT}}$ includes:
\begin{itemize}
    \item All recovered assignments \( (i'_t, j'_t) \) for \( t \in T \), and
    \item The missed item \( (i, j) \).
\end{itemize}

\noindent $\ALG$ includes:
\begin{itemize}
    \item All recovered assignments \( (i'_t, j'_t) \) for \( t \in T \).
\end{itemize}

\noindent Since all assignments in $\mathbf{ALG}$ are also included in $\overline{\mathbf{OPT}}$, and $\overline{\mathbf{OPT}}$ additionally includes \( (i, j) \), we conclude:
\[
\Delta w_{j'}(\mathbf{ALG}) \leq \Delta w_{j'}(\overline{\mathbf{OPT}}) \quad \text{for all } j'.
\]

\paragraph{Value-Based Substitution: (Lines~\ref{line:large-case-b} to~\ref{line:large-case-b-done}).}
Let \( t^* \in T \) be the index minimizing \( v_{i'_t}^{\OPT} \), as chosen on Line~\ref{line:large-argmin}. These notations follow the same convention as in the algorithm.

\noindent $\overline{\mathbf{OPT}}$ includes:
\begin{itemize}
    \item All recovered assignments \( (i'_t, j'_t) \) for \( t \in T \), and
    \item The missed item \( (i, j) \).
\end{itemize}

\noindent $\mathbf{ALG}$ includes:
\begin{itemize}
    \item All recovered assignments \( (i'_t, j'_t) \) for \( t \in T \setminus \{t^*\} \), and
    \item The substitute assignment \( (i'_{t^*}, j) \).
\end{itemize}

\noindent Observe that every assignment in $\mathbf{ALG}$ is also present in $\overline{\mathbf{OPT}}$, except that $\mathbf{ALG}$ replaces \( (i, j) \in \overline{\mathbf{OPT}} \) with \( (i'_{t^*}, j) \). By \AlternativeExists, 
we have:
\[
w_{i'_{t^*}, j} \leq w_{i, j}.
\]
All other assignments are shared between $\mathbf{ALG}$ and $\overline{\mathbf{OPT}}$, so for every knapsack \( j' \), we conclude:
\[
\Delta w_{j'}(\mathbf{ALG}) \leq \Delta w_{j'}(\overline{\mathbf{OPT}}) \quad \text{for all } j'.
\]

\paragraph{Exact Substitution: (Lines~\ref{line:large-final-loop} to~\ref{line:large-final-loop-done}).}
In the final cleanup phase, both $\mathbf{ALG}$ and $\overline{\mathbf{OPT}}$ add exactly the same set of assignments $\{(i, j')\}$ for each \( i \in S_{j,k}^{\mathrm{queried}} \). Therefore,
\[
\Delta w_{j'}(\mathbf{ALG}) = \Delta w_{j'}(\bar{\mathbf{OPT}}) \quad \text{for all } j'.
\]

\smallskip
In all cases, for every knapsack \( j' \), we have
\[
\Delta w_{j'}(\mathbf{ALG}) \leq \Delta w_{j'}(\overline{\mathbf{OPT}}),
\]
completing the proof.
\end{proof}

\begin{proof}[Proof of Lemma~\ref{lem:fill-small-less-capacity}]
Consider a specific invocation of \textsc{FillSmallBucket} for bucket $(j,k)$ and recall that $S_{j,k}^{\mathrm{queried}} \subseteq \bar{B}_{j,k} \subseteq Q \cap R$.
By construction, both $\overline{\mathbf{OPT}}$ and $\mathbf{ALG}$ assigns items exclusively from $\bar{B}_{j,k}$, and no item is assigned multiple times at this step. Since the aggregate bucket sets $\{\bar{B}_{j,k}\}_{j,k}$ are pairwise disjoint  (Section~\ref{sec:notation}), we guarantee that no item is assigned more than once across the entire reconstruction process. Furthermore, the inclusion $\bar{B}_{j,k} \subseteq Q \cap R$ ensures that $\mathbf{ALG}$ relies solely on available queried active elements.

We now proceed to bound the capacity consumption. We analyze the weight increments $\Delta w_{j'}(\cdot)$ -- representing the total weight added to knapsack $j'$ during this specific call -- by examining the two distinct branches in Algorithm~\ref{alg:fill-small} separately.

\paragraph{Exact Substitution: (Line~\ref{line:small-if-easy}) \( S_{j,k}^{\mathrm{missed}} = \emptyset \).}
In this case, all items in \( S_{j,k}^{\mathrm{queried}} \) are assigned to both $\overline{\mathbf{OPT}}$ and $\mathbf{ALG}$ using the same knapsack assignments from $\mathbf{OPT}$. Hence, for every knapsack \( j' \),
\[
\Delta w_{j'}(\mathbf{ALG}) = \Delta w_{j'}(\overline{\mathbf{OPT}}).
\]

\paragraph{Density-Based Substitution: (Line~\ref{line:small-if-not-easy}) \( S_{j,k}^{\mathrm{missed}} \neq \emptyset \).}
The algorithm constructs a substitution set $S \subseteq \bar{B}_{j,k}$ satisfying $\sum_{i \in S} w_{ij} \leq \sum_{i \in S_{j,k}^{\mathrm{missed}}} w_{ij}$.

\noindent $\overline{\mathbf{OPT}}$ includes:
\begin{itemize}
\item All queried items $S_{j,k}^{\mathrm{queried}}$ with their original knapsack assignments from $\mathbf{OPT}$.
\item All missed items $S_{j,k}^{\mathrm{missed}}$ assigned to knapsack $j$.
\end{itemize}
$\ALG$ includes:
\begin{itemize}
\item Queried items not in $S$ with their original assignments.
\item All items in $S$ reassigned to knapsack $j$.
\end{itemize}

For each knapsack $j'$, if $j' \neq j$, then $\mathbf{ALG}$ assigns only a subset of the items that $\overline{\mathbf{OPT}}$ assigns to $j'$. On the other hand, if $j' = j$, then $\ALG$ makes the assignment in a way that 
$$\Delta w_j(\mathbf{ALG}) = \sum_{i \in S} w_{ij} \leq \sum_{i \in S_{j,k}^{\mathrm{missed}}} w_{ij} \leq \Delta w_j(\overline{\mathbf{OPT}}).$$
Therefore, $\Delta w_{j'}(\mathbf{ALG}) \leq \Delta w_{j'}(\overline{\mathbf{OPT}})$ for all $j'$.
\end{proof}

\begin{proof}[Proof of Lemma~\ref{lem:fill-super-less-capacity}]
Consider the execution of \textsc{FillAllSuperBuckets}. By construction, the sets $S^{\mathrm{queried}}_{j,K+1}$ and $S^{\mathrm{missed}}_{j,K+1}$ consist exclusively of active items in the super bucket $\bar B_{j,K+1}$, and these sets are disjoint from all other bucket sets. Since the aggregate bucket partition $\{\bar B_{j,k}\}$ is pairwise disjoint across all $(j,k)$, and each item is removed from $S^{\mathrm{queried}}_{j,K+1}$ or $S^{\mathrm{missed}}_{j,K+1}$ immediately after being processed, no item is ever assigned more than once to either $\mathbf{ALG}$ or $\overline{\mathbf{OPT}}$ throughout the reconstruction. In addition, every item assigned to $\mathbf{ALG}$ comes from $S^{\mathrm{queried}}_{j,K+1} \subseteq \bar{B}_{j,K+1} \subseteq Q \cap R$ (Section~\ref{sec:notation}).

We now compare the weight increments to each knapsack. For every $j$ and for every $i \in S^{\mathrm{queried}}_{j,K+1}$, the procedure adds $(i,j')$ to both $\overline{\mathbf{OPT}}$ and $\mathbf{ALG}$, where $(i,j')$ is the assignment of item $i$ in the optimal solution. Thus, every weight increase incurred by $\mathbf{ALG}$ due to queried super items is exactly matched by a corresponding increase in $\overline{\mathbf{OPT}}$.

For missed super items, i.e., $i \in S^{\mathrm{missed}}_{j,K+1}$, the procedure adds $(i,j)$ to $\overline{\mathbf{OPT}}$ but does not add anything to $\mathbf{ALG}$.

Combining these observations, we conclude that for every knapsack $j' \in [m]$,
\[
\Delta w_{j'}(\mathbf{ALG}) \le \Delta w_{j'}(\overline{\mathbf{OPT}}).
\]
\end{proof}

\subsection{Value Analysis}
\begin{proof}[Proof of Lemma~\ref{lem:fill-large-value}]
For a fixed call to \textsc{FillLargeBucket}, we consider each case in Algorithm~\ref{alg:fill-large} separately. 
Recall that \( \Delta v(\cdot) \) denotes the total value added to all knapsacks during this call.

\paragraph{Direct Substituition: (Lines~\ref{line:large-if-available} to~\ref{line:large-case0-done}).}
In this case, $\overline{\mathbf{OPT}}$ adds $(i, j)$ and $\mathbf{ALG}$ adds $(i', j)$. Since $i' \in \bar{B}_{j,k}$, by bucket definition, we have
\[
\Delta v(\mathbf{ALG}) = v_{i'j} \geq \frac{1}{1 + \epsilon^2} \cdot v_{ij} \ge (1 - \epsilon^2) \cdot \Delta v(\overline{\mathbf{OPT}}).
\]

The following two cases are the most critical and complex scenarios. An intuitive explanation is provided in Figure~\ref{fig:large-substitution-cases-example}.
\begin{figure}[!htb]
    \centering
    \fbox{%
        \begin{minipage}{\textwidth}
    \begin{subfigure}{\textwidth}
        \centering
        \hspace*{0.15cm}
        \scalebox{0.8}{
            \begin{tikzpicture}
                    \draw[dashed, rounded corners] (11-15,0.9+0.2) rectangle (14.5-15,2.3+0.2);
                    \fill[fill=green] (12.3-15,1.5+0.2) rectangle (13.2-15,2.05+0.2);
                    \draw[draw=blue,thick] (12.3-15,1.5+0.2) rectangle (13.2-15,2.2+0.2);
                    \draw [decorate, decoration={brace, amplitude=6pt}]
                        (12.3-15,1.5+0.2) -- (12.3-15,2.05+0.2)
                        node[midway, left=6pt] {\(v_i^{\OPT}\)};
                    \draw [decorate, decoration={brace, amplitude=6pt, mirror}]
                        (13.2-15,1.5+0.2) -- (13.2-15,2.2+0.2)
                        node[midway, right=6pt] {\(v_{ij}\)};
                    \draw [decorate, decoration={brace, amplitude=6pt, mirror}]
                        (12.3-15,1.5+0.2) -- (13.2-15,1.5+0.2)
                        node[midway, below=5pt] {\(w_{ij}\)};
                        
                \node at (-2, 0.55) {\textbf{\(\overline{\mathbf{OPT}}\)} adds \fcolorbox{black}{green}{\textbf{green}} :};
        
                \draw (0,0) -- (8,0);
                
                \foreach \x in {0,2,4,6,8} {
                \draw (\x,0.1) -- (\x,-0.1);
                }
                \foreach \x in {2,4,6,8} {
                \pgfmathtruncatemacro{\label}{\x/2}
                \node[below] at (\x-1,-0.05) {\(\bar{B}_{j,k,\label}\)};
                }
                
                \fill[green] (1-0.25, 0) rectangle (1+0.25, 1.8);
                \draw[draw=blue, thick] (1-0.25, 0) rectangle (1+0.25, 1.4);
                \node[] at (1, 0.55) {\(i'_{1}\)};
                \fill[green] (3-0.3, 0) rectangle (3+0.3, 1.45);
                \draw[draw=blue, thick] (3-0.3, 0) rectangle (3+0.3, 1.6);
                \node[] at (3, 0.55) {\(i'_{2}\)};
                \fill[green] (5-0.35, 0) rectangle (5+0.35, 1.55);
                \draw[draw=blue, thick] (5-0.35, 0) rectangle (5+0.35, 1.4);
                \node[] at (5, 0.55) {\(i'_{3}\)};
                \fill[green] (7-0.4, 0) rectangle (7+0.4, 1.85);
                \draw[draw=blue, thick] (7-0.4, 0) rectangle (7+0.4, 1.5);
                \node[] at (7, 0.55) {\(i'_{4}\)};
        
                \fill[green] (9-0.5, 0) rectangle (9+0.5, 1.35);
                \draw[draw=blue, thick] (9-0.5, 0) rectangle (9+0.5, 1.35);
                \node[blue] at (9, 0.55) {\(i\)};
                \node[below] at (9,-0.05) {\( i \in S^{\mathrm{missed}}_{j,k} \)};
                \node[red, text width=3cm, align=center] at (9, 2.1) {Smallest \( v_i^{\OPT} \) \\ in the group};
                \draw[->, red, thick] (9,1.7) to (9, 1.35);
        
                \draw[draw=red] (0-0.2,1.35) -- (9.75+0.6, 1.35);
            \end{tikzpicture}
        }

        \hspace*{0.15cm}
        \scalebox{0.8}{
            \begin{tikzpicture}
                \draw[dashed, rounded corners] (11-15,0.9+0.2) rectangle (14.5-15,2.3+0.2);
                \fill[fill=yellow] (12.3-15,1.5+0.2) rectangle (13.2-15,2.05+0.2);
                \draw[draw=blue,thick] (12.3-15,1.5+0.2) rectangle (13.2-15,2.2+0.2);
                \draw [decorate, decoration={brace, amplitude=6pt}]
                    (12.3-15,1.5+0.2) -- (12.3-15,2.05+0.2)
                    node[midway, left=6pt] {\(v_i^{\ALG}\)};
                \draw [decorate, decoration={brace, amplitude=6pt, mirror}]
                    (13.2-15,1.5+0.2) -- (13.2-15,2.2+0.2)
                    node[midway, right=6pt] {\(v_{ij}\)};
                \draw [decorate, decoration={brace, amplitude=6pt, mirror}]
                    (12.3-15,1.5+0.2) -- (13.2-15,1.5+0.2)
                    node[midway, below=5pt] {\(w_{ij}\)};
                    
            \node at (-2, 0.55) {\textbf{\(\mathbf{ALG}\)} adds \fcolorbox{black}{yellow}{\textbf{yellow}} :};
    
            \draw (0,0) -- (8,0);
            
            \foreach \x in {0,2,4,6,8} {
            \draw (\x,0.1) -- (\x,-0.1);
            }
            \foreach \x in {2,4,6,8} {
            \pgfmathtruncatemacro{\label}{\x/2}
            \node[below] at (\x-1,-0.05) {\(\bar{B}_{j,k,\label}\)};
            }
            
            \fill[yellow] (1-0.25, 0) rectangle (1+0.25, 1.8);
            \draw[draw=blue, thick] (1-0.25, 0) rectangle (1+0.25, 1.4);
            \node[] at (1, 0.55) {\(i'_{1}\)};
            \fill[yellow] (3-0.3, 0) rectangle (3+0.3, 1.45);
            \draw[draw=blue, thick] (3-0.3, 0) rectangle (3+0.3, 1.6);
            \node[] at (3, 0.55) {\(i'_{2}\)};
            \fill[yellow] (5-0.35, 0) rectangle (5+0.35, 1.55);
            \draw[draw=blue, thick] (5-0.35, 0) rectangle (5+0.35, 1.4);
            \node[] at (5, 0.55) {\(i'_{3}\)};
            \fill[yellow] (7-0.4, 0) rectangle (7+0.4, 1.85);
            \draw[draw=blue, thick] (7-0.4, 0) rectangle (7+0.4, 1.5);
            \node[] at (7, 0.55) {\(i'_{4}\)};
    
            \draw[draw=blue, thick] (9-0.5, 0) rectangle (9+0.5, 1.35);
            \node[blue] at (9, 0.55) {\(i\)};
            \node[below] at (9,-0.05) {\( i \in S^{\mathrm{missed}}_{j,k} \)};
    
            \draw[draw=red] (0-0.2,1.35) -- (9.75+0.6, 1.35);
            
            \node[text width=3cm, align=center] at (9, 2.1) {$\mathbf{ALG}$ ignores \textcolor{red}{i}};
            \draw[->, thick] (9,1.85) to (9, 1.35);
        \end{tikzpicture}
        }
        \caption*{\textbf{Value-Based Rejection}: Green rectangles represent \(\mathbf{OPT}\)'s total value gain (hence \(\overline{\mathbf{OPT}}\)'s), blue rectangles show each item's value when placed in knapsack \(j\), and yellow rectangles represent \(\mathbf{ALG}\)'s total value gain. The missed item \(i\) has the smallest value in \(\overline{\mathbf{OPT}}\), so \(\mathbf{ALG}\) ignores \(i\) while retaining other allocations. Notably, \(v_i^{\OPT}\) constitutes at most \(\frac{1}{\alpha}\) of \(\overline{\mathbf{OPT}}\)'s value gain.
        }
    \end{subfigure}

    \vspace{1em}

    \begin{subfigure}{\textwidth}
        \centering
        \scalebox{0.8}{
            \begin{tikzpicture}
                \draw[dashed, rounded corners] (11-15,0.9+0.2) rectangle (14.5-15,2.3+0.2);
                \fill[fill=green] (12.3-15,1.5+0.2) rectangle (13.2-15,2.05+0.2);
                \draw[draw=blue,thick] (12.3-15,1.5+0.2) rectangle (13.2-15,2.2+0.2);
                \draw [decorate, decoration={brace, amplitude=6pt}]
                    (12.3-15,1.5+0.2) -- (12.3-15,2.05+0.2)
                    node[midway, left=6pt] {\(v_i^{\OPT}\)};
                \draw [decorate, decoration={brace, amplitude=6pt, mirror}]
                    (13.2-15,1.5+0.2) -- (13.2-15,2.2+0.2)
                    node[midway, right=6pt] {\(v_{ij}\)};
                \draw [decorate, decoration={brace, amplitude=6pt, mirror}]
                    (12.3-15,1.5+0.2) -- (13.2-15,1.5+0.2)
                    node[midway, below=5pt] {\(w_{ij}\)};
                    
            \node at (-2, 0.55) {\textbf{\(\overline{\mathbf{OPT}}\)} adds \fcolorbox{black}{green}{\textbf{green}} :};
    
            \draw (0,0) -- (8,0);
            
            \foreach \x in {0,2,4,6,8} {
            \draw (\x,0.1) -- (\x,-0.1);
            }
            \foreach \x in {2,4,6,8} {
            \pgfmathtruncatemacro{\label}{\x/2}
            \node[below] at (\x-1,-0.05) {\(\bar{B}_{j,k,\label}\)};
            }
            
            \fill[green] (1-0.25, 0) rectangle (1+0.25, 1.8);
            \draw[draw=blue, thick] (1-0.25, 0) rectangle (1+0.25, 1.4);
            \node[] at (1, 0.55) {\(i'_{1}\)};
            \fill[green] (3-0.3, 0) rectangle (3+0.3, 1.45);
            \draw[draw=blue, thick] (3-0.3, 0) rectangle (3+0.3, 1.6);
            \node[] at (3, 0.55) {\(i'_{2}\)};
            \fill[green] (5-0.35, 0) rectangle (5+0.35, 1.25);
            \draw[draw=blue, thick] (5-0.35, 0) rectangle (5+0.35, 1.4);
            \node[red] at (5, 0.55) {\(i'_{3}\)};
            \node[red, text width=3cm, align=center] at (5, 2.1) {Smallest \( v_i^{\OPT} \) \\ in the group};
            \draw[->, red, thick] (5,1.7) to (5, 1.25);
            \fill[green] (7-0.4, 0) rectangle (7+0.4, 1.85);
            \draw[draw=blue, thick] (7-0.4, 0) rectangle (7+0.4, 1.5);
            \node[] at (7, 0.55) {\(i'_{4}\)};
    
            \fill[green] (9-0.5, 0) rectangle (9+0.5, 1.35);
            \draw[draw=blue, thick] (9-0.5, 0) rectangle (9+0.5, 1.35);
            \node[blue] at (9, 0.55) {\(i\)};
            \node[below] at (9,-0.05) {\( i \in S^{\mathrm{missed}}_{j,k} \)};
    
            \draw[draw=red] (0-0.2,1.25) -- (9.75+0.6, 1.25);
        \end{tikzpicture}
        }
        \scalebox{0.8}{
            \begin{tikzpicture}
                \draw[dashed, rounded corners] (11-15,0.9+0.2) rectangle (14.5-15,2.3+0.2);
                \fill[fill=yellow] (12.3-15,1.5+0.2) rectangle (13.2-15,2.05+0.2);
                \draw[draw=blue,thick] (12.3-15,1.5+0.2) rectangle (13.2-15,2.2+0.2);
                \draw [decorate, decoration={brace, amplitude=6pt}]
                    (12.3-15,1.5+0.2) -- (12.3-15,2.05+0.2)
                    node[midway, left=6pt] {\(v_i^{\ALG}\)};
                \draw [decorate, decoration={brace, amplitude=6pt, mirror}]
                    (13.2-15,1.5+0.2) -- (13.2-15,2.2+0.2)
                    node[midway, right=6pt] {\(v_{ij}\)};
                \draw [decorate, decoration={brace, amplitude=6pt, mirror}]
                    (12.3-15,1.5+0.2) -- (13.2-15,1.5+0.2)
                    node[midway, below=5pt] {\(w_{ij}\)};
                    
            \node at (-2, 0.55) {\textbf{\(\mathbf{ALG}\)} adds \fcolorbox{black}{yellow}{\textbf{yellow}} :};
    
            \draw (0,0) -- (8,0);
            
            \foreach \x in {0,2,4,6,8} {
            \draw (\x,0.1) -- (\x,-0.1);
            }
            \foreach \x in {2,4,6,8} {
            \pgfmathtruncatemacro{\label}{\x/2}
            \node[below] at (\x-1,-0.05) {\(\bar{B}_{j,k,\label}\)};
            }
            
            \fill[yellow] (1-0.25, 0) rectangle (1+0.25, 1.8);
            \draw[draw=blue, thick] (1-0.25, 0) rectangle (1+0.25, 1.4);
            \node[] at (1, 0.55) {\(i'_{1}\)};
            \fill[yellow] (3-0.3, 0) rectangle (3+0.3, 1.45);
            \draw[draw=blue, thick] (3-0.3, 0) rectangle (3+0.3, 1.6);
            \node[] at (3, 0.55) {\(i'_{2}\)};
            \fill[yellow] (5-0.35, 0) rectangle (5+0.35, 1.4);
            \draw[draw=blue, thick] (5-0.35, 0) rectangle (5+0.35, 1.4);
            \node[red] at (5, 0.55) {\(i'_{3}\)};
            \fill[yellow] (7-0.4, 0) rectangle (7+0.4, 1.85);
            \draw[draw=blue, thick] (7-0.4, 0) rectangle (7+0.4, 1.5);
            \node[] at (7, 0.55) {\(i'_{4}\)};
    
            \draw[draw=blue, thick] (9-0.5, 0) rectangle (9+0.5, 1.35);
            \node[blue] at (9, 0.55) {\(i\)};
            \node[below] at (9,-0.05) {\( i \in S^{\mathrm{missed}}_{j,k} \)};
    
            \draw[draw=red] (0-0.2,1.25) -- (9.75+0.6, 1.25);
        
            \draw[->, bend left=45, thick] (5,1.4) to (9, 1.4);
            
            \node[above=7pt, align=center] at ($(5,1.9)!0.5!(9,1.9)$) 
              {\textcolor{black}{$\mathbf{ALG}$ uses \textcolor{red}{\( i'_3 \)} to substitute \textcolor{blue}{\( i \)}}};
        \end{tikzpicture}
        }
        \caption*{\textbf{Value-Based Substitution}: In this case, the potential substitute \(i'_3\) has the smallest value in \(\overline{\mathbf{OPT}}\), so \(\mathbf{ALG}\) reassigns \(i'_3\) to knapsack \(j\) to substitute missed item \(i\) while retaining other allocations. Note that \(v_{i'_3}^{\OPT}\) constitutes at most \(\frac{1}{\alpha}\) of \(\overline{\mathbf{OPT}}\)'s value gain, and \(v_{i'_3j}\) nearly covers the entire \(v_{ij}\) value.}
    \end{subfigure}

    \caption{
    \textbf{Visualization of substitution in \textsc{FillLargeBucket}(\(\overline{\mathbf{OPT}}, \mathbf{ALG}, j, k\)).}
    Assume \(\alpha = 5\) and \( T = [\alpha - 1] = 4 \). For each \( t \in T \), let \( (i'_t, j'_t) \in \mathbf{OPT} \) be the selected potential substitution from bucket \( \bar{B}_{j,k,t} \). The upper subfigure visualizes \emph{Value-Based Rejection} case and the lower subfigure demonstrates \emph{Value-Based Substitution}.
    }
    \label{fig:large-substitution-cases-example}
    \end{minipage}
    }
\end{figure}

\paragraph{Value-Based Rejection (Lines~\ref{line:large-case-a} to~\ref{line:large-case-a-done}).}


We analyze this case under the condition that the Excess Weight Events $\mathcal{E}_{j,k,t}$ hold for all $t \in \{1, \ldots, \alpha-1\}$. 
Since this case arises only when $S_{j,k}^{\mathrm{missed}} \neq \emptyset$, 
we may condition on this event and apply Lemma~\ref{lem:bucket_weight_guarantee}, 
which ensures that each $\mathcal{E}_{j,k,t}$ occurs with probability at least $1 - \epsilon^{2}$. 
Moreover, because the buckets $\bar{B}_{j,k,t}$ for $t = 1, \ldots, \alpha-1$ consist of disjoint item sets, 
the events $\mathcal{E}_{j,k,t}$ are independent, and thus the joint probability that all such events hold is at least $(1 - \epsilon^{2})^{\alpha}$. 
Conditioned on these events, Lemma~\ref{lem:missed-properties} (\textsc{SizeMatch}) guarantees that each queried bucket $\bar{B}_{j,k,t}$ contains a suitable substitute item not already selected by $\mathbf{ALG}$.

In this case, the loop on Line~\ref{line:large-t-loop} selects a substitute from every bucket, so we define \( T = [\alpha - 1] \). For each \( t \in T \), let \( (i'_t, j'_t) \in \mathbf{OPT} \) denote the assignment recovered from bucket \( \bar{B}_{j,k,t} \).

\noindent $\overline{\mathbf{OPT}}$ includes:
\begin{itemize}
    \item All recovered assignments \( (i'_t, j'_t) \) for \( t \in T \), and
    \item The missed item \( (i, j) \).
\end{itemize}

\noindent $\mathbf{ALG}$ includes:
\begin{itemize}
    \item All recovered assignments \( (i'_t, j'_t) \) for \( t \in T \).
\end{itemize}

\noindent The value gained by $\overline{\mathbf{OPT}}$ is:
\[
\Delta v(\overline{\mathbf{OPT}}) = \sum_{t \in T} v_{i'_t}^{\OPT} + v_i^{\OPT},
\]
and the value gained by $\mathbf{ALG}$ is:
\[
\Delta v(\mathbf{ALG}) = \sum_{t \in T} v_{i'_t}^{\OPT}.
\]
By the algorithm’s condition, the missed item \( i \) has the smallest value among the \( \alpha \) items assigned by $\overline{\mathbf{OPT}}$, which implies:
\[
\Delta v(\mathbf{ALG}) \geq \left(1 - \frac{1}{\alpha} \right) \cdot \Delta v(\overline{\mathbf{OPT}}).
\]

Taking expectation over the Excess Weight Events, we conclude:
\[
\E_R[\Delta v(\mathbf{ALG})] \geq (1 - \epsilon^2)^\alpha \cdot \left(1 - \frac{1}{\alpha} \right) \cdot \Delta v(\overline{\mathbf{OPT}}).
\]

\paragraph{Value-Based Substitution: (Lines~\ref{line:large-case-b} to~\ref{line:large-case-b-done}).}
We continue under the same conditioning that all Excess Weight Events \( \mathcal{E}_{j,k,t} \) hold for \( t = 1 \) to \( \alpha - 1 \), which occurs with probability at least \( (1 - \epsilon^2)^\alpha \). As in the previous case, let \( T = [\alpha - 1] \), and for each \( t \in T \), let \( (i'_t, j'_t) \in \mathbf{OPT} \) be the recovered assignment from bucket \( \bar{B}_{j,k,t} \), following the algorithm's notation. Let \( t^* \in T \) denote the index minimizing \( v_{i'_t}^{\OPT} \), as chosen on Line~\ref{line:large-argmin}.

\noindent $\overline{\mathbf{OPT}}$ includes:
\begin{itemize}
    \item All recovered assignments \( (i'_t, j'_t) \) for \( t \in T \), and
    \item The missed item \( (i, j) \).
\end{itemize}

\noindent $\mathbf{ALG}$ includes:
\begin{itemize}
    \item All assignments \( (i'_t, j'_t) \) for \( t \in T \setminus \{t^*\} \), and
    \item The substitute assignment \( (i'_{t^*}, j) \).
\end{itemize}

\noindent The value gained by $\overline{\mathbf{OPT}}$ is:
\[
\Delta v(\overline{\mathbf{OPT}}) = \sum_{t \in T} v_{i'_t}^{\OPT} + v_i^{\OPT},
\]
and the value gained by $\mathbf{ALG}$ is:
\[
\Delta v(\mathbf{ALG}) = \sum_{t \in T \setminus \{t^*\}} v_{i'_t}^{\OPT} + v_{i'_{t^*}, j}.
\]

Since \( i'_{t^*} \in \bar{B}_{j,k} \), we have \( v_{i'_{t^*}, j} \geq (1 - \epsilon^2) \cdot v_i^{\OPT} \), hence
\[
\Delta v(\mathbf{ALG}) \geq (1 - \epsilon^2) \cdot \left( \sum_{t \in T \setminus \{t^*\}} v_{i'_t}^{\OPT} + v_i^{\OPT} \right).
\]

Moreover, by the algorithm’s condition, the item \( i'_{t^*} \) has the smallest value among the \( \alpha \) items assigned by \( \overline{\mathbf{OPT}} \). Hence,
\[
\sum_{t \in T \setminus \{t^*\}} v_{i'_t}^{\OPT} + v_i^{\OPT} \geq \left(1 - \frac{1}{\alpha} \right) \cdot \Delta v(\overline{\mathbf{OPT}}),
\]
which leads to:
\[
\Delta v(\mathbf{ALG}) \geq (1 - \epsilon^2) \cdot \left(1 - \frac{1}{\alpha} \right) \cdot \Delta v(\overline{\mathbf{OPT}}).
\]

Taking expectation over the Excess Weight Events, we conclude:
\[
\E_R[\Delta v(\mathbf{ALG})] \geq (1 - \epsilon^2)^\alpha \cdot \left(1 - \frac{1}{\alpha} - \epsilon^2 \right) \cdot \Delta v(\overline{\mathbf{OPT}}).
\]

\paragraph{Exact Subsititution: (Lines~\ref{line:large-final-loop} to~\ref{line:large-final-loop-done}).}
In the final cleanup phase, both $\mathbf{ALG}$ and $\overline{\mathbf{OPT}}$ add exactly the same set of assignments $\{(i, j')\}$ for each \( i \in S_{j,k}^{\mathrm{queried}} \). Therefore,
\[
\Delta v(\mathbf{ALG}) = \Delta v(\overline{\mathbf{OPT}}).
\]

\paragraph{Conclusion.}
Combining all cases, we have
\[
\E_R[\Delta v(\mathbf{ALG})] \geq (1 - \epsilon^2)^\alpha \cdot \left(1 - \frac{1}{\alpha} - \epsilon^2 \right) \cdot \Delta v(\overline{\mathbf{OPT}}).
\]
Substituting \( \alpha = 1/\epsilon \) and using \( (1 - \epsilon^2)^\alpha \geq 1 - \epsilon \), we get
\[
\E_R[\Delta v(\mathbf{ALG})] \geq (1 - 2\epsilon) \cdot \Delta v(\overline{\mathbf{OPT}}),
\]
completing the proof.
\end{proof}

Before proving the approximation guarantees of \textsc{FillSmallBucket}, we first formalize a structural property of the prefix \( S \).

\begin{claim} [Ordering of Value-to-Weight Ratios]
\label{clm:small-ratio-ordering}
Let \( S \subseteq \bar{B}_{j,0} \) be the prefix selected in Line~\ref{line:small-define-s} of Algorithm~\ref{alg:fill-small}.  
Then, for all \( i \in S \) and all \( i' \in \bar{B}_{j,0} \setminus S \), we have:
\[
\frac{v_i^{\OPT}}{w_{ij}} \leq \frac{v_{ij}}{w_{ij}} \quad \text{and} \quad
\frac{v_i^{\OPT}}{w_{ij}} \leq \frac{v_{i'}^{\OPT}}{w_{i'j}}.
\]
\end{claim}

\begin{proof}
By construction, the first inequality is exactly the condition under which
items are selected into \( S \).
Moreover, the items in \( \bar{B}_{j,0} \) are processed in non-decreasing
order of their value-to-weight ratio \( \frac{v_i^{\OPT}}{w_{ij}} \), and the
prefix \( S \) is chosen according to this ordering.
Consequently, every item in \( S \) has a value-to-weight ratio under
\( \mathbf{OPT} \) that is no larger than that of any item not in \( S \).
\end{proof}

We now proceed to the proof of the approximation guarantees of \textsc{FillSmallBucket}.

\begin{proof}[Proof of Lemma~\ref{lem:fill-small-value}]
We analyze the approximation guarantee by considering the two branches of Algorithm~\ref{alg:fill-small}, where $\Delta v(\cdot)$ denotes the total value added during this call.

\paragraph{Exact Substitution: (Line~\ref{line:small-if-easy}).}
If \( S_{j,k}^{\mathrm{missed}} = \emptyset \), then both \( \mathbf{ALG} \) and \( \overline{\mathbf{OPT}} \) add exactly the same set of assignments \( S_{j,k}^{\mathrm{queried}} \). Hence,
\[
\Delta v(\mathbf{ALG}) = \Delta v(\overline{\mathbf{OPT}}).
\]

The case below is particularly critical and involve detailed argument. Figure~\ref{fig:small-substitution-cases-example} offers an intuitive explanation of this scenario.
\input{figure_small}
\paragraph{Density-Based Substitution: (Line~\ref{line:small-define-s}) \( S_{j,k}^{\mathrm{missed}} \neq \emptyset \).}
We condition on the Excess Weight Events $\mathcal{E}_{j,k,t}$ holding for all $t$. Under these events, Lemma~\ref{lem:missed-properties} ensures sufficient substitutes exist in each queried bucket $\bar{B}_{j,0,t}$. Let $S \subseteq \bar{B}_{j,0}$ be the prefix selected in Line~\ref{line:small-define-s}, and $i^*$ be the first excluded item.

\begin{itemize}
\item[(1)] If \( w_j(S) + w_{i^*j} \geq w_j(S_{j,k}^{\mathrm{missed}}) \): \\

\paragraph{Value decomposition.}
By construction of the algorithm, we have:
\[
\Delta v(\mathbf{ALG}) + v_{i^*j} = \sum_{i \in S \cup \{i^*\}} v_{ij} + \sum_{i \in \bar{B}_{j,0} \setminus S} v_i^{\OPT},
\]
\[
\Delta v(\overline{\mathbf{OPT}}) = \sum_{i \in S} v_i^{\OPT} + \sum_{i \in \bar{B}_{j,0} \setminus S} v_i^{\OPT} + \sum_{i \in S_{j,k}^{\mathrm{missed}}} v_{ij}.
\]

\paragraph{Density-based bound.}
We now apply Lemma~\ref{lem:missed-properties} (\textsc{AlternativeExists}) under the Excess Weight Events \( \mathcal{E}_{j,0,t} \). 
It ensures that for any \( i \in S_{j,k}^{\mathrm{missed}} \) and any \( i \in S \cup \{i^*\} \),
\[
\frac{v_{i'j}}{w_{i'j}} \ge \frac{v_{ij}}{w_{ij}}.
\]

Let \( d_{\max} := \max_{i \in S_{j,k}^{\mathrm{missed}}} \frac{v_{ij}}{w_{ij}} \). Then:
\[
\sum_{i \in S \cup \{i^*\}} v_{ij} 
= \sum_{i \in S \cup \{i^*\}} \frac{v_{ij}}{w_{ij}} \cdot w_{ij}
\ge d_{\max} \cdot \left( w_j(S) + w_{i^*j} \right)
\ge d_{\max} \cdot w_j(S_{j,k}^{\mathrm{missed}}),
\]
\[
\sum_{i \in S_{j,k}^{\mathrm{missed}}} v_{ij}
= \sum_{i \in S_{j,k}^{\mathrm{missed}}} \frac{v_{ij}}{w_{ij}} \cdot w_{ij}
\le d_{\max} \cdot w_j(S_{j,k}^{\mathrm{missed}}).
\]
Thus,
\[
\sum_{i \in S \cup \{i^*\}} v_{ij} \ge \sum_{i \in S_{j,k}^{\mathrm{missed}}} v_{ij},
\]
and therefore,
\[
\Delta v(\mathbf{ALG}) + v_{i^*j} \ge \sum_{i \in \bar{B}_{j,0} \setminus S} v_i^{\OPT} + \sum_{i \in S_{j,k}^{\mathrm{missed}}} v_{ij}.
\]

\paragraph{Ordering-based bound.}
By construction, the items in \( \bar{B}_{j,0} \) are sorted in increasing order of \( \frac{v_i^{\OPT}}{w_{ij}} \). Since \( i^* \) is the first item not selected in \( S \), we know:
\[
\max_{i \in S} \frac{v_i^{\OPT}}{w_{ij}} \le \frac{v_{i^*}^{\OPT}}{w_{i^*j}} = \min_{i \in \bar{B}_{j,0} \setminus S} \frac{v_i^{\OPT}}{w_{ij}}.
\]
Define \( d_{i^*} := \frac{v_{i^*}^{\OPT}}{w_{i^*j}} \). 

Under the Excess Weight Events, we know:
\[
w_j(\bar{B}_{j,0}) \ge (\alpha - 1) \cdot C_j,
\]
and since \( S_{j,k}^{\mathrm{missed}} \subseteq \OPT \), we have:
\[
w_j(S_{j,k}^{\mathrm{missed}}) \le C_j,
\quad
w_j(S) \le w_j(S_{j,k}^{\mathrm{missed}}) \le C_j,
\]
so:
\[
w_j(\bar{B}_{j,0} \setminus S) \ge (\alpha - 2) \cdot C_j.
\]

Now,
\[
\sum_{i \in \bar{B}_{j,0} \setminus S} v_i^{\OPT} 
= \sum_{i \in \bar{B}_{j,0} \setminus S} \frac{v_i^{\OPT}}{w_{ij}} \cdot w_{ij}
\ge d_{i^*} \cdot w_j(\bar{B}_{j,0} \setminus S)
\ge d_{i^*} \cdot (\alpha - 2) \cdot C_j,
\]
\[
\sum_{i \in S} v_i^{\OPT}
= \sum_{i \in S} \frac{v_i^{\OPT}}{w_{ij}} \cdot w_{ij}
\le d_{i^*} \cdot w_j(S)
\le d_{i^*} \cdot C_j.
\]
Thus,
\[
\sum_{i \in \bar{B}_{j,0} \setminus S} v_i^{\OPT} \ge (\alpha - 2) \cdot \sum_{i \in S} v_i^{\OPT}.
\]
Since each \( i \in S \) satisfies \( v_{ij} > v_i^{\OPT} \), it follows that:
\[
\sum_{i \in S} v_{ij} \ge \sum_{i \in S} v_i^{\OPT}.
\]

Combining, we get:
\[
\Delta v(\mathbf{ALG}) + v_{i^*j}
> \sum_{i \in S} v_{ij} + \sum_{i \in \bar{B}_{j,0} \setminus S} v_i^{\OPT}
\ge \sum_{i \in S} v_i^{\OPT} + (\alpha - 2) \cdot \sum_{i \in S} v_i^{\OPT}
= (\alpha - 1) \cdot \sum_{i \in S} v_i^{\OPT}.
\]

\paragraph{Value recombination.}
From above, we have:
\[
\Delta v(\mathbf{ALG}) + v_{i^*j} \ge \sum_{i \in \bar{B}_{j,0} \setminus S} v_i^{\OPT} + \sum_{i \in S_{j,k}^{\mathrm{missed}}} v_{ij},
\]
and
\[
\Delta v(\mathbf{ALG}) + v_{i^*j} \ge (\alpha - 1) \cdot \sum_{i \in S} v_i^{\OPT}.
\]
Multiplying the first inequality by \( \frac{\alpha - 1}{\alpha} \) and the second by \( \frac{1}{\alpha} \), and adding:
\[
\Delta v(\mathbf{ALG}) + v_{i^*j}
\ge \frac{\alpha - 1}{\alpha} \cdot \left( \sum_{i \in \bar{B}_{j,0} \setminus S} v_i^{\OPT} + \sum_{i \in S_{j,k}^{\mathrm{missed}}} v_{ij} \right)
+ \frac{1}{\alpha} \cdot (\alpha - 1) \cdot \sum_{i \in S} v_i^{\OPT}.
\]
This implies:
\[
\Delta v(\mathbf{ALG}) + v_{i^*j} \ge \frac{\alpha - 1}{\alpha} \cdot \Delta v(\overline{\mathbf{OPT}}).
\]

\paragraph{Final adjustment.}
By definition of the bucket, we know \( v_{i^*j} \le \epsilon^2 \cdot M_j \), so we conclude:
\[
\Delta v(\mathbf{ALG}) \ge \left(1 - \tfrac{1}{\alpha} \right) \cdot \Delta v(\overline{\mathbf{OPT}}) - \epsilon^2 \cdot M_j.
\]
This completes the proof for subcase (1).

\item[(2)] If \( w_j(S) + w_{i^*j} < w_j(S_{j,k}^{\mathrm{missed}}) \): \\

\paragraph{Value decomposition.} 
By construction of the algorithm, we have:
\[
\Delta v(\mathbf{ALG}) = \sum_{i \in S} v_{ij} + \sum_{i \in \bar{B}_{j,0} \setminus S} v_i^{\OPT},
\]
\[
\Delta v(\overline{\mathbf{OPT}}) = \sum_{i \in S} v_i^{\OPT} + \sum_{i \in \bar{B}_{j,0} \setminus S} v_i^{\OPT} + \sum_{i \in S_{j,k}^{\mathrm{missed}}} v_{ij}.
\]

Since each \( i \in S \) satisfies \( v_{ij} > v_i^{\OPT} \), we obtain:
\[
\Delta v(\mathbf{ALG}) \ge \sum_{i \in S} v_i^{\OPT} + \sum_{i \in \bar{B}_{j,0} \setminus S} v_i^{\OPT}.
\]

\paragraph{Density-based bound.} 
Let \( d_{\max} := \max_{i \in S_{j,k}^{\mathrm{missed}}} \frac{v_{ij}}{w_{ij}} \). 
By construction, for any \( i \in S_{j,k}^{\mathrm{missed}} \) and any \( i' \in \bar{B}_{j,0} \setminus S \), we have:
\[
\frac{v_{i'}^{\OPT}}{w_{i'j}} 
\ge 
\frac{v_{i^*}^{\OPT}}{w_{i^*j}} 
\ge 
\frac{v_{i^*j}}{w_{i^*j}} 
\ge 
\frac{v_{ij}}{w_{ij}}.
\]
The second inequality holds because \( i^* \notin S \), and by our assumption that \( w_j(S) + w_{i^*j} < w_j(S_{j,k}^{\mathrm{missed}}) \), its exclusion is not due to weight constraints. Rather, it must be that \( v_{i^*}^{\OPT} \ge v_{i^*j} \).

Also, under the Excess Weight Event, Lemma~\ref{lem:missed-properties} (\textsc{AlternativeExists}) ensures that for any \( i \in S_{j,k}^{\mathrm{missed}} \) and any \( i' \in S \), 
\[
\frac{v_{i'j}}{w_{i'j}} \ge \frac{v_{ij}}{w_{ij}}.
\]
Therefore, the value of \( \mathbf{ALG} \) satisfies:
\[
\begin{aligned}
\Delta v(\mathbf{ALG})
&= \sum_{i \in S} v_{ij} + \sum_{i \in \bar{B}_{j,0} \setminus S} v_i^{\OPT} \\
&= \sum_{i \in S} \frac{v_{ij}}{w_{ij}} \cdot w_{ij} + \sum_{i \in \bar{B}_{j,0} \setminus S} \frac{v_i^{\OPT}}{w_{ij}} \cdot w_{ij} \\
&\ge \sum_{i \in S} d_{\max} \cdot w_{ij} + \sum_{i \in \bar{B}_{j,0} \setminus S} d_{\max} \cdot w_{ij} \\
&= d_{\max} \cdot w_j(\bar{B}_{j,0}) \\
&\ge d_{\max} \cdot (\alpha - 1) C_j,
\end{aligned}
\]
where the final inequality holds under the assumption that all Excess Weight Events \( \mathcal{E}_{j,0,t} \) hold.

\paragraph{Bounding missed value.}
\[
\begin{aligned}
\sum_{i \in S_{j,k}^{\mathrm{missed}}} v_{ij}
&= \sum_{i \in S_{j,k}^{\mathrm{missed}}} \frac{v_{ij}}{w_{ij}} \cdot w_{ij} \\
&\le \sum_{i \in S_{j,k}^{\mathrm{missed}}} d_{\max} \cdot w_{ij} \\
&= d_{\max} \cdot w_j(S_{j,k}^{\mathrm{missed}}) \\
&\le d_{\max} \cdot C_j,
\end{aligned}
\]
since \( S_{j,k}^{\mathrm{missed}} \subseteq \OPT \) and the total weight packed by \( \OPT \) into knapsack \( j \) must be at most \( C_j \).

Combining the two bounds, we get:
\[
\Delta v(\mathbf{ALG}) \ge (\alpha - 1) \cdot \sum_{i \in S_{j,k}^{\mathrm{missed}}} v_{ij}.
\]

\paragraph{Value recombination.}
From above:
\[
\Delta v(\mathbf{ALG}) > \sum_{i \in S} v_i^{\OPT} + \sum_{i \in \bar{B}_{j,0} \setminus S} v_i^{\OPT},
\]
and
\[
\Delta v(\mathbf{ALG}) \ge (\alpha - 1) \cdot \sum_{i \in S_{j,k}^{\mathrm{missed}}} v_{ij}.
\]
We multiply the first inequality by \( \frac{\alpha - 1}{\alpha} \), and the second by \( \frac{1}{\alpha} \), and then add both sides:
\[
\frac{\alpha - 1}{\alpha} \cdot \Delta v(\mathbf{ALG}) + \frac{1}{\alpha} \cdot \Delta v(\mathbf{ALG})
> \frac{\alpha - 1}{\alpha} \cdot \left( \sum_{i \in S} v_i^{\OPT} + \sum_{i \in \bar{B}_{j,0} \setminus S} v_i^{\OPT} \right)
+ \frac{1}{\alpha} \cdot (\alpha - 1) \cdot \sum_{i \in S_{j,k}^{\mathrm{missed}}} v_{ij}.
\]
This implies:
\[
\Delta v(\mathbf{ALG}) \ge \frac{\alpha - 1}{\alpha} \cdot \left( \sum_{i \in S} v_i^{\OPT} + \sum_{i \in \bar{B}_{j,0} \setminus S} v_i^{\OPT} + \sum_{i \in S_{j,k}^{\mathrm{missed}}} v_{ij} \right)
= \frac{\alpha - 1}{\alpha} \cdot \Delta v(\overline{\mathbf{OPT}}).
\]

This completes the proof for subcase (2).

\end{itemize}

In both subcases, we have:
\[
\Delta v(\mathbf{ALG}) \geq \left(1 - \tfrac{1}{\alpha} \right) \cdot \Delta v(\overline{\mathbf{OPT}}) - \epsilon^2 \cdot M_j,
\]
provided that all Excess Weight Events \( \mathcal{E}_{j,k,t} \) for \( t = 1, \dots, \alpha \) hold simultaneously. 
Since each such event occurs independently with probability at least \( 1 - \epsilon^2 \), the joint probability is at least \( (1 - \epsilon^2)^\alpha \). Therefore, in expectation:
\[
\E_R[\Delta v(\mathbf{ALG})] \geq (1 - \epsilon^2)^\alpha \cdot \left( \left(1 - \tfrac{1}{\alpha} \right) \cdot \Delta v(\overline{\mathbf{OPT}}) - \epsilon^2 \cdot M_j \right).
\]

\paragraph{Conclusion.}
Combining both cases, and substituting \( \alpha = 1/\epsilon \), noting that \( (1 - \epsilon^2)^\alpha \geq 1 - \epsilon \), we conclude:
\[
\E_R[\Delta v(\mathbf{ALG})] \geq (1 - 2\epsilon) \cdot \Delta v(\overline{\mathbf{OPT}}) - \epsilon^2 \cdot M_j,
\]
as claimed.
\end{proof}

We now proceed to the approximation guarantee of \textsc{FillAllSuperBuckets}.

\begin{proof}[Proof of Lemma~\ref{lem:fill-super-value}]
Let $\Delta_v(\bar{\mathbf{OPT}})$ and $\Delta_v(\mathbf{ALG})$ denote the changes in value of the current optimal solution $\bar{\mathbf{OPT}}$ and the algorithm's solution $\mathbf{ALG}$, respectively, during the execution of \textsc{FillAllSuperBuckets}.
By construction,
\[
\mathbb{E}_R[\Delta v(\overline{\mathbf{OPT}})]
-
\mathbb{E}_R[\Delta v(\mathbf{ALG})]
=
\mathbb{E}_R\!\left[\sum_{j=1}^{m}\sum_{i\in S^{\mathrm{missed}}_{j,K+1}} v_{ij}\right].
\]

We define the fixed set
\[
J := \{\, j\in[m] : \exists (i,j)\in B_{j,K+1}\ \text{with } i\notin Q \ \text{and } w_{ij} \le C_j \,\},
\]
which is fixed because Algorithm~\ref{alg:query_set_construction} is deterministic and non-adaptive. Then

\begin{align*}
\mathbb{E}_R[\Delta v(\overline{\mathbf{OPT}})] - 
\mathbb{E}_R[\Delta v(\mathbf{ALG})]
&= \textstyle
\mathbb{E}_R\!\left[\sum_{j\in J}\sum_{i\in S^{\mathrm{missed}}_{j,K+1}} v_{ij}\right] \\
&\le \textstyle
\mathbb{E}_R\!\left[\sum_{j\in J}\sum_{(i,j)\in\mathbf{OPT}} v_{ij}\right] \\
&= \textstyle
\sum_{j\in J} M_j.
\end{align*}

\noindent
The first equality uses $S^{\mathrm{missed}}_{j,K+1}=\emptyset$ for $j\notin J$ (Section~\ref{sec:notation}), the inequality follows from $S^{\mathrm{missed}}_{j,K+1}\subseteq\{\,i:(i,j)\in\mathbf{OPT}\,\}$, and the final equality uses linearity of expectation and $M_j=\mathbb{E}_R[v(\mathbf{OPT}_j)]$, where $v(\mathbf{OPT}_j)=\sum_{(i,j)\in\mathbf{OPT}} v_{ij}$.
Thus,
\begin{equation}
\label{eq:opt_alg_diff}
\mathbb{E}_R[\Delta v(\overline{\mathbf{OPT}})]
-
\mathbb{E}_R[\Delta v(\mathbf{ALG})]
\le
\sum_{j\in J} M_j.    
\end{equation}

We now upper bound $\sum_{j\in J}M_j$ in terms of $\mathbb{E}_R[v(\mathbf{OPT})]$.  
Fix any realization $R$ and any $j\in J$. For each round $t$, let $\hat{B}_{j,K+1,t}$ denote the set of all items (active or inactive) that Algorithm~\ref{alg:query_set_construction} queries through bucket $(j,K+1)$ in that round. By construction, the families $\{\hat{B}_{j,K+1,t}\}_{j,t}$ are pairwise disjoint across different choices of $j$ and $t$. Because $j\in J$, bucket $(j,K+1)$ always exhausts its query capacity, so by Lemma~\ref{lem:bucket_weight_guarantee},
\[
\mathcal{E}_{j,K+1,t}:=\{ w_j(\bar{B}_{j,K+1,t})\ge C_j\}
\quad\text{holds with probability }1-\epsilon^2.
\]
Fix $t=1$.  Conditional on $\mathcal{E}_{j,K+1,1}$, there exists an active item $i\in\bar{B}_{j,K+1,1}\subseteq\hat{B}_{j,K+1,1}$ such that $v_{ij}>M_j/\epsilon$.  
If $(i,j)\in\mathbf{OPT}$ then $v^{\mathrm{OPT}}_i=v_{ij}>M_j/\epsilon$.  
If $(i,j)\notin\mathbf{OPT}$ and $v^{\mathrm{OPT}}_i+v(\mathbf{OPT}_j)\le M_j/\epsilon$, then replacing $\mathbf{OPT}_j$ by assigning $i$ improves the objective, contradicting optimality.  
Hence
\[
v_i^{\mathrm{OPT}} + v(\mathbf{OPT}_j) > M_j/\epsilon
\quad\text{conditional on }\mathcal{E}_{j,K+1,1}.
\]
Summing over $i\in\hat{B}_{j,K+1,1}$ gives
\[
\sum_{i\in\hat{B}_{j,K+1,1}} v_i^{\mathrm{OPT}} + v(\mathbf{OPT}_j)
> M_j/\epsilon
\quad\text{conditional on }\mathcal{E}_{j,K+1,1}.
\]
If $\mathcal{E}_{j,K+1,1}$ fails, worst case none of the items are active, so
\[
\sum_{i\in\hat{B}_{j,K+1,1}} v_i^{\mathrm{OPT}} + v(\mathbf{OPT}_j)\ge 0.
\]

Using these two cases,
\begin{align*}
    \mathbb{E}_R\!&\left[\sum_{i\in\hat{B}_{j,K+1,1}} v_i^{\mathrm{OPT}}
+ v(\mathbf{OPT}_j)\right]\\
    &=\mathbb{E}_R\left[\sum_{i\in\hat{B}_{j,K+1,1}} v_i^{\mathrm{OPT}}
+ v(\mathbf{OPT}_j) \mid \mathcal{E}_{j,K+1,1}\right] \cdot \Pr[\mathcal{E}_{j,K+1,1}]\\
    &\quad+\mathbb{E}_R\left[\sum_{i\in\hat{B}_{j,K+1,1}} v_i^{\mathrm{OPT}} + v(\mathbf{OPT}_j) \mid \neg\mathcal{E}_{j,K+1,1}]\right] \cdot \Pr[\neg\mathcal{E}_{j,K+1,1}]\\
    &\geq (1-\epsilon^2)\cdot \frac{M_j}{\epsilon}.
\end{align*}
    

Summing over all $j\in J$ and using linearity of expectation, we obtain
\begin{align*}
    \sum_{j\in J}
        &\mathbb{E}_R\!\left[
        \sum_{i\in\hat{B}_{j,K+1,1}} v_i^{\OPT}
        +
        v(\mathbf{OPT_j})
        \right]\\
    &= \sum_{j\in J}
\mathbb{E}_R\!\left[\sum_{i\in\hat{B}_{j,K+1,1}} v_i^{\OPT}\right]
+
\sum_{j\in J}
\mathbb{E}_R\!\left[v(\mathbf{OPT_j})\right]
\le 
2 \cdot \E_R[v(\mathbf{OPT})].
\end{align*}
Here, the inequality follows from the following fact. The sets $\hat{B}_{j,K+1,1}$ are pairwise disjoint across $j$, the
items contributing to 
$\sum_{i\in\hat{B}_{j,K+1,1}} v_i^{\OPT}$ 
never overlap between different knapsacks.  
Likewise, the sets $\mathbf{OPT_j}$ are also disjoint across $j$, so the terms
$v(\mathbf{OPT_j})$ do not overlap either.
Therefore, each of the two sums is individually upper bounded by $\mathbb{E}_R[v(\mathbf{OPT})]$.

Hence, we obtain
\[
2\,\mathbb{E}_R[v(\mathbf{OPT})]
\ge
(1-\epsilon^2)\sum_{j\in J}\frac{M_j}{\epsilon},
\]
implying that
\[
\sum_{j\in J} M_j
\le
\frac{2\epsilon}{1-\epsilon^2}\,\mathbb{E}_R[v(\mathbf{OPT})].
\]
Finally, for $\epsilon\le \tfrac12$, we have
$\frac{2\epsilon}{1-\epsilon^2} \le 3\epsilon$, yielding
\[
\sum_{j\in J} M_j
\le
3\epsilon \cdot \mathbb{E}_R[v(\mathbf{OPT})].
\]

Substituting this into Equation~\eqref{eq:opt_alg_diff} proves that
\[
\mathbb{E}_R[\Delta v(\overline{\mathbf{OPT}})]
-
\mathbb{E}_R[\Delta v(\mathbf{ALG})]
\le
3\epsilon \cdot \mathbb{E}_R[v(\mathbf{OPT})].
\]
\end{proof}

\subsection{Optimal Value Preservation}

\begin{proof}[Proof of Lemma~\ref{lem:baropt-equals-opt}]
The reconstruction algorithm iterates over all bucket indices $(j,k)$ in Algorithm~\ref{alg:reconstruction}.  
For each pair $(j,k\le K)$, it invokes either \textsc{FillSmallBucket} (Algorithm~\ref{alg:fill-small}) or \textsc{FillLargeBucket} (Algorithm~\ref{alg:fill-large}).  
After completing all these buckets, it invokes \textsc{FillAllSuperBuckets} (Algorithm~\ref{alg:fill-super}) to process the super bucket $(j,K+1)$ for every knapsack~$j$.

Each call of \textsc{FillSmallBucket} or \textsc{FillLargeBucket} processes all items in the sets $S_{j,k}^{\mathrm{queried}}$ and $S_{j,k}^{\mathrm{missed}}$, and adds to $\overline{\mathbf{OPT}}$ exactly those pairs $(i,j')$ with $(i,j')\in\mathbf{OPT}$.  
In other words, every such call contributes only true optimal assignments to $\overline{\mathbf{OPT}}$ and ensures that no item is ever included more than once.

Similarly, \textsc{FillAllSuperBuckets} processes all remaining optimal assignments in the super buckets.  
For each $j$, every item in $S^{\mathrm{queried}}_{j,K+1}$ or $S^{\mathrm{missed}}_{j,K+1}$ corresponds to some $(i,j')\in\mathbf{OPT}$, and the procedure inserts exactly this assignment once and only once into $\overline{\mathbf{OPT}}$.

Moreover, the sets $S_{j,k}^{\mathrm{queried}}$ and $S_{j,k}^{\mathrm{missed}}$ partition the items in $\OPT$ (Section~\ref{sec:notation}). Therefore, by the end of the reconstruction process, $\overline{\mathbf{OPT}}$ contains exactly the same item-knapsack assignments as $\mathbf{OPT}$, i.e., $\overline{\mathbf{OPT}} = \mathbf{OPT}$.
\end{proof}

\section{Conclusion}

We introduced a polyhedral framework for sparsification that extends beyond uniform structures such as matching and matroids to capacity-constrained problems including knapsack, multiple knapsack, and the generalized assignment problem. Our results demonstrate that despite the inherent hardness of these problems, one can construct $(1-\epsilon)$-approximate sparsifiers with degree polynomial in $1/p$ and $1/\epsilon$, independent of the problem size. This establishes a clean separation between optimization complexity and sparsification complexity: while exact or near-exact optimization remains intractable, identifying a small query set that preserves optimality up to $(1-\epsilon)$ is efficiently possible.

More broadly, our work highlights sparsification as a lens for rethinking stochastic combinatorial optimization. The polyhedral notion of degree captures structural redundancy without relying on cardinality, suggesting applications far beyond knapsack-type problems. A central open question remains: can we design size-independent sparsifiers for general integer linear programs, with degree depending only on $1/p$, $1/\epsilon$, and intrinsic structural parameters? Progress on this front would push the boundary of query-efficient optimization and clarify the fundamental role of sparsification in stochastic combinatorial optimization.

\subsection{Further Related Works}
\label{sec:related_works}
\paragraph{Knapsack and GAP Problem Approximations}
The Knapsack problem, proven NP-complete by Karp \cite{karp1972reducibility}, has driven extensive research into fully polynomial-time approximation schemes (FPTAS). For background on knapsack approximations, see the monographs
\emph{Knapsack Problems} by Kellerer et al. \cite{kellerer2004knapsack}
and \emph{The Design of Approximation Algorithms} by Williamson and Shmoys \cite{williamson2011design}. The first published FPTAS for Knapsack was given by Ibarra and Kim \cite{ibarra1975fptas}, with running time $\tilde{O}(n + (1/\epsilon)^4)$, where $\tilde{O}$ suppresses polylogarithmic factors in $n$ and $1/\epsilon$. Recent breakthroughs include Chen et al. \cite{chen2024nearly} and Mao \cite{mao2024approximation}, who achieved FPTAS algorithms running in $O(n + (1/\epsilon)^2)$ time and established hardness results for subquadratic FPTAS under the (min, +)-convolution hypothesis. The General Assignment Problem (GAP) and Multiple Knapsack Problem (MKP) present additional complexity: Chekuri and Khanna \cite{chekuri2005mkp} showed that MKP with even two knapsacks does not admit an FPTAS and that GAP is APX-hard, while designing a PTAS for MKP. For GAP, the leading approximation guarantee remains the $1 - 1/e + \varepsilon$ ratio obtained by Feige and Vondr\'{a}k \cite{feige2006allocation}, where $\varepsilon>0$ is an absolute constant. 
On the hardness side, Chakrabarty and Goel \cite{chakrabarty2010gapbound} demonstrated that any polynomial-time algorithm achieving a factor strictly better than $10/11$ would imply $\mathrm{P}=\mathrm{NP}$.

\paragraph{Stochastic Matching}
Stochastic matching, a fundamental special case of stochastic packing problems, was pioneered by Blum et al. \cite{blum2015ignorance} for the unweighted $2$-set packing problem, achieving a $(1-\epsilon)/2$ approximation with degree $O(1/p^{1/\epsilon})$. This initiated a sequence of improvements by a series of follow-up works~\cite{assadi2019towards, assadi2019stochastic, behnezhad2020stochastic, behnezhad2019stochastic, dughmi2023sparsification}. For weighted stochastic matching, the current best result achieves a $0.68$ approximation \cite{derakhshan2025weighted}, improving upon the previous $0.536$ bound \cite{dughmi2023sparsification}. Most recently, Azarmehr
et al. \cite{azarmehr2025} achieved a breakthrough $(1-\epsilon)$-approximation sparsifier for unweighted stochastic matching with poly$(1/p)$ degree on general graphs using Local Computation Algorithms. 

\paragraph{General Stochastic Packing Problems}
The theoretical foundation for general stochastic packing problems was established by Maehara and Yamaguchi \cite{yamaguchi2018stochastic}, who introduced a unified framework for stochastic packing integer programs with query-efficient algorithms. Their approach provided non-adaptive sparsifiers for various additive SPPs including the GAP and sparse integer linear programs, though with degree dependent on the number of constraints. They extended this work to stochastic monotone submodular maximization \cite{maehara2019stochastic}, developing query strategies for problems with uniform exchange properties. Dughmi et al. \cite{dughmi2023sparsification} addressed the constraint-dependence limitation by employing contention resolution schemes (CRS) that yield sparsifiers with degree independent of the number of constraints, achieving degree polynomial only in $1/p$ and $1/\epsilon$ for matroids and weighted matching.

\paragraph{Set Selection Under Explorable Uncertainty}
A related line of research examines optimization problems under explorable uncertainty, where parameter values are initially hidden within known uncertainty intervals and can be revealed through costly queries. Megow and Schl\"oter
\cite{megow2023set} study the set selection problem, which seeks to identify the subset with minimum total value among a given family of sets. In contrast, Schl\"oter \cite{schloter2025complexity} investigates the knapsack problem under explorable uncertainty, where the goal is to compute a minimal query set sufficient for determining an optimal knapsack solution. This work establishes strong theoretical barriers, proving the problem is $\Sigma_p^2$-complete and showing that no non-trivial approximation exists unless $\Sigma_p^2 = \Delta_p^2$. To circumvent these limitations, the author develops algorithms for resource-augmented variants that compute approximate solutions while competing against optimal query sets. Both works measure algorithmic performance through query complexity -- the number of uncertain parameters that must be revealed to solve the underlying optimization problem.

\section{Acknowledgments}
This paper is based upon work supported by the Air Force Office of Scientific Research under award number FA9550-24-1-0261. Any opinions, findings, and conclusions or recommendations expressed in this document are those of the authors and do not necessarily reflect the views of the United States Air Force.

\newpage
\bibliographystyle{plainurl}
\bibliography{references}

\newpage
\appendix

\section{Missing Proofs}

\subsection{Proof of Lemma~\ref{lem:bucket_weight_guarantee}}
\label{app:weight-concentration}

\begin{proof}
Let 
\(
X := \sum_{i \in S} w_i \cdot \mathbf{1}_{\{i \in R\}}
\)
be the total active weight in $S$, where each item becomes active independently with probability~$p$.  
Define the normalized random variable
\[
Y := \frac{X}{C} = \sum_{i \in S} \frac{w_i}{C}\cdot \mathbf{1}_{\{i \in R\}}.
\]
Since every weight satisfies $0 \le w_i \le C$, each summand of $Y$ lies in the interval $[0,1]$, and hence $Y$ is suitable for applying the multiplicative Chernoff bound.

Let $\mu := \mathbb{E}[Y] = \mathbb{E}[X]/C$.  
The lemma assumption 
\[
\sum_{i\in S} w_i \ge \frac{\tau(\eps)}{p} \cdot C
\]
implies 
\[
\mu = \frac{p\sum_{i\in S} w_i}{C} \;\ge\; \tau(\eps) \; > 1.
\]
Note that
\[
\{ X < C \} \iff \{ Y < 1 \}.
\]

Define 
\[
\delta := 1 - \frac{1}{\mu} \in (0,1),
\qquad\text{so that}\qquad
1 = (1-\delta)\mu.
\]
Applying the multiplicative Chernoff bound for sums of independent $[0,1]$-bounded variables, we obtain
\[
\Pr[X < C]
= \Pr[Y < 1]
= \Pr\!\left[Y < (1-\delta)\mu\right]
\le \exp\!\left(-\frac{\delta^{2}}{2}\,\mu\right)
= \exp\!\left(-\frac{(\mu - 1)^2}{2\mu}\right).
\]

Solving the inequality 
\[
\frac{(\mu-1)^2}{2\mu} \ge \ln(1/\eps)
\]
for the smallest feasible value of $\mu$ gives the threshold
\[
\tau(\eps)
:= 1 + \ln(1/\eps)
   + \sqrt{\ln^2(1/\eps) + 2\ln(1/\eps)}.
\]

Since $\mu \ge \tau(\eps)$, the Chernoff bound yields
\[
\Pr[X < C]
\le \exp\!\left(-\frac{(\tau(\eps)-1)^2}{2\tau(\eps)}\right)
= \eps.
\]
Finally,
\[
\tau(\eps) = \Theta(\ln(1/\eps)),
\]
completing the proof.
\end{proof}

\section{Additional Experimental Details}
\label{appendix:experiments}

\subsection{Motivation for the Experimental Study}
\label{app:motivation}
While our theoretical analysis focuses on sparsifiers for ``stochastic'' knapsack and GAP formulations, many real-world applications are deterministic and predominantly rely on ``deterministic'' integer linear programming (ILP) solvers. This raises a natural empirical question: can our sparsification technique serve as an effective preprocessing step that reduces problem size and accelerates solution times in deterministic settings while preserving near-optimal objective values?

\paragraph{Experimental Overview.}
To address this question, we conduct a comprehensive empirical evaluation structured as follows. We begin by examining the characteristics and patterns observed across diverse data sources, then detail our synthetic instance generation methodology. Subsequently, we establish our evaluation metrics and experimental protocol before presenting and analyzing the computational results.

\subsection{Data Sources and Instance Generation}
\label{app:data}

We evaluated two widely-used public GAP benchmarks: the OR-Library benchmark by Beasley \cite{beasley_dataset} and the repository by Yagiura et al. \cite{dataset2}. 
However, sparsification proves ineffective on both datasets because the instances are already highly sparse -- all items can fit within the knapsack constraints when capacity is doubled. This inherent sparsity renders our sparsification technique redundant, as there are few (or no) items to remove without significantly impacting solution quality.

Given these limitations of existing benchmarks, we turn to synthetic instance generation, which enables systematic control over problem characteristics and allows us to create meaningfully dense instances where sparsification can demonstrate its effectiveness.

\begin{table}[h]
\centering
\begin{threeparttable}
\caption{Parameter grid for synthetic instance generation.}
\label{tab:param-grid}
\begin{tabular}{ll}
\toprule
Parameter & Values \\
\midrule
Item count $n$ & $\{1000, 2000, 5000, 10000^\}$ \\
Bin count $m$ & $\{1, 2, 5\}$ \\
Correlation $\rho$ & $\{-0.8,-0.5,-0.3,0,0.3,0.5,0.8\}$ \\
Redundancy targets & $\{1,2,3,5,8,13,22,36,60,100\}$ \\
Marginal pairs $(F_v,F_w)$ & $\{\textsc{Uniform},\textsc{TruncNormal}\}^2$ \\
Trials per setting & $8$ i.i.d.\ replicates \\
\bottomrule
\end{tabular}
\end{threeparttable}
\end{table}

The full generation parameters are summarized in Table~\ref{tab:param-grid}. 
For each configuration we generate a GAP \emph{maximization} instance with $n$ items and $m$ bins. 
Values $v_{ij}$ and weights $w_{ij}$ are drawn from fixed one–dimensional marginals and \emph{coupled via a Gaussian copula} with parameter $\rho$: 
\[
(Z_1,Z_2)\sim \mathcal{N}\!\bigl(\mathbf{0},\,\Sigma_\rho\bigr),\qquad
\Sigma_\rho=\begin{pmatrix}1 & \rho\\ \rho & 1\end{pmatrix},\qquad
U_\ell=\Phi(Z_\ell)\ (\ell\in\{1,2\}),
\]
and we set
\[
(v_{ij},w_{ij})=\bigl(F_v^{-1}(U_1),\,F_w^{-1}(U_2)\bigr).
\]
This preserves the marginals while varying the concordance (rank dependence) through $\rho\in\{-0.8,-0.5,-0.3,0,0.3,0.5,0.8\}$. 
We use 
\[
F_v \in \{\mathrm{Uniform}[0,100],\ \left.\mathcal{N}(50,15^2)\right|_{[0,100]}\},\qquad
F_w \in \{\mathrm{Uniform}[1,20],\ \left.\mathcal{N}(10,5^2)\right|_{[1,30]}\},
\]
so that the uniform marginals provides a flat baseline, while truncated Gaussians yield bounded, concentrated coefficients; a tiny offset $10^{-2}$ is added to all weights to avoid exact ties. 
Capacities are scaled by a target redundancy parameter via
\[
C_j \;=\; Q_{0.05}\!\left(\{w_{ij}\}_{i=1}^n\right)\cdot \frac{n}{m}\cdot \frac{1}{\text{redundancy target}},
\]
where $Q_{0.05}$ is the empirical $5\%$ quantile and the targets follow a roughly geometric (log-scale) progression to span low–high regimes compactly. 
These settings range from relatively easy to harder ILP regimes and cover negative through positive value–weight dependence. 
Each configuration is replicated independently for 8 trails, and all experiments use a fixed master seed with deterministic burn-in for exact reproducibility.

\subsection{Sparsifier (LP-Driven Variant)}
\label{app:sparsifier}

We instantiate the bucket sparsifier from Algorithm~\ref{alg:query_set_construction}; see also Corollary~\ref{cor:global-oracle-sparsifier}.  
The hyperparameters used in our deterministic experiments are summarized in Table~\ref{tab:sparsifier-hyper}.  

\begin{table}[h]
\centering
\caption{Hyperparameters for the LP-driven sparsifier.}
\label{tab:sparsifier-hyper}
\begin{tabular}{ll}
\toprule
Parameter & Setting \\
\midrule
$p$ (stochasticity) & $1$ (deterministic regime) \\
$\alpha$ (rounds) & $1$ (single pass) \\
$\tau$ (per-bucket allowance) & $1$ (nominal capacity budget) \\
$\epsilon$ (resolution) & $0.2$ \\
$M_j$ (global scale) & $\mathrm{LP_{OPT}}$ for all $j \in [m]$ \\
$K$ (bucket count) & $\left\lceil \tfrac{1}{\epsilon^2}\log \tfrac{1}{\epsilon^2}\right\rceil$ \\
\bottomrule
\end{tabular}
\end{table}

Table~\ref{tab:sparsifier-hyper} highlights three design choices.
First, since randomness plays no role in the deterministic setting, we fix $p=1$ and adopt the simplest scheme of one round with nominal allowance ($\alpha=\tau=1$).
Second, we replace the oracle scale $M_j$ from Corollary~\ref{cor:global-oracle-sparsifier} with a single global scale given by the LP optimum on the full instance, counted in preprocessing time.
Third, for the bucket count we adopt $K=\lceil \tfrac{1}{\epsilon^2}\log \tfrac{1}{\epsilon^2}\rceil$, which reduces overhead compared to the corollary-level bound. We deliberately set a relatively large $\epsilon=0.2$ so that $K$ is small, allowing us to test whether such coarse bucketing is still sufficient to yield high approximation in practice.

Given these settings, bucketing and item selection follow Algorithm~\ref{alg:query_set_construction} verbatim, and the restricted ILP on query set $Q$ is then solved to termination, thereby testing whether $Q$ retains a near-optimal solution. All preprocessing (LP solve $+$ bucketing) is counted as part of the sparsification pipeline, so that we can assess whether the sparsifier serves as a genuine preprocessing step that accelerates solving GAP.

\subsection{Solver Setup and Time-Fair Protocol}
\label{app:solver}

\paragraph{Environment.}
All experiments are implemented in \texttt{Julia} using \texttt{JuMP} with the \texttt{Gurobi} backend, with default solver parameters. Wall-clock time is measured externally around solver calls.

\paragraph{Three runs per instance.}
For each generated instance we perform three runs under identical solver settings:
\begin{enumerate}
\item \emph{Exact solving}: solve the full ILP to termination, recording the optimum $\OPT_{\text{full}}$, the runtime $T_{\text{full}}$, and the realized item count.
\item \emph{Sparsification followed by exact solving}: apply the bucket sparsifier from Section~\ref{app:sparsifier} to obtain a query set $Q$, then solve the restricted ILP to termination. The end-to-end runtime -- including the LP relaxation, sparsification, and restricted ILP -- is denoted $T_{\text{sparse}}$, and the attained optimum by $\OPT_{\text{sparse}}$.
\item \emph{Time-constrained exact solving}: run the full ILP under a wall-clock budget of $T_{\text{sparse}}$, yielding the best incumbent $\OPT_{\text{full}}^{(\text{cut})}$. This serves as the equal-time benchmark for performance comparison.
\end{enumerate}

\subsection{Evaluation Metrics and Notation}
\label{app:metrics}

Having specified instance generation (Section~\ref{app:data}) and the three-run protocol (Section~\ref{app:solver}), we now define the metrics used throughout our analysis, summarized in Table~\ref{tab:metrics}.

\begin{table}[h]
\centering
\caption{Evaluation metrics.}
\label{tab:metrics}
\begin{tabular}{ll}
\toprule
\textbf{Metric} & \textbf{Definition} \\
\midrule
Approximation ratio 
& $\mathrm{OPT}_{\text{sparse}} / \mathrm{OPT}_{\text{full}}$ \\[0.3em]

Speedup 
& $T_{\text{full}} / T_{\text{sparse}}$ (end-to-end time) \\[0.3em]

Equal-time value ratio (ETR) 
& $\mathrm{OPT}_{\text{sparse}} / \mathrm{OPT}_{\text{full}}^{(\text{cut})}$, where \\
& $\mathrm{OPT}_{\text{full}}^{(\text{cut})}$ is the best incumbent within $T_{\text{sparse}}$ \\[0.3em]

Redundancy ratio $r$ 
& $n / |\OPT_{\text{full}}|$, where $|\OPT_{\text{full}}|$ is the number of items \\
& selected by the exact solving ILP solution \\[0.3em] 

\bottomrule
\end{tabular}
\end{table}

The approximation ratio measures value preservation, speedup captures end-to-end runtime gains, and ETR provides a time-fair benchmark against the full ILP.  
The redundancy ratio serves as a structural parameter that links the experiments back to our theoretical analysis.

\subsection{Analysis of Experimental Results}
\label{app:analysis}

We are ready to present the empirical findings.

\begin{table}[h]
\centering
\begin{tabular}{lccc}
\toprule
Setting & Approx.\ ratio & Speedup & ETR \\
\midrule
All runs          & $0.998$ & $2.79\times$ & $2.76\times$ \\
$r>4,\, m>1,\, n>1000$     & ---     & $4.67\times$ & $5.14\times$ \\
\bottomrule
\end{tabular}

\captionsetup{justification=centering}
\caption{Aggregate and filtered averages.\label{tab:agg-averages}}
\captionsetup{justification=raggedright}
\end{table}

Table~\ref{tab:agg-averages} summarizes the macro results: the algorithm using sparsification preserves a high approximation ratio and yields end-to-end speedups, with markedly stronger gains on the informative slice that restricts to instances with \(r>4\), \(m>1\), and \(n>1000\). Guided by these results, we next examine how performance varies with the problem size $n$, the number of bins $m$, the redundancy ratio $r$, and the correlation parameter $\rho$.

\paragraph{Speedup versus redundancy.}

\begin{figure}[h]
    \centering
    \includegraphics[width=\linewidth]{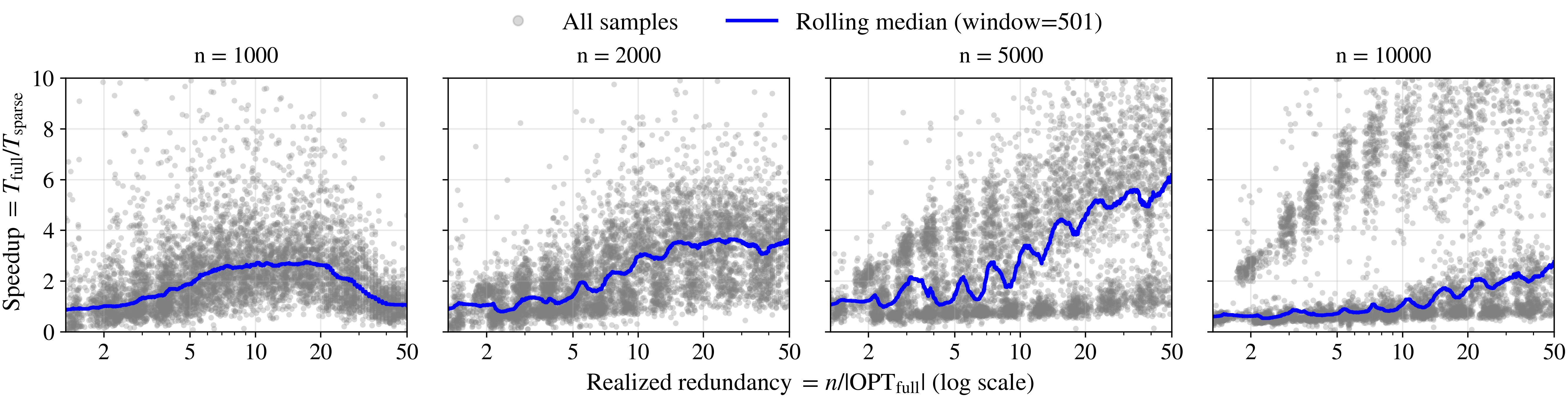}
    \caption{End-to-end speedup ($T_{\mathrm{full}}/T_{\mathrm{sparse}}$) versus realized redundancy $r=n/|\mathrm{OPT}_{\mathrm{full}}|$ for $n\in\{1000,2000,5000, 1000\}$. 
    Each column fixes $n$ (labeled at the top). Gray points are individual runs; the blue curve is a rolling median computed with window size $501$. 
    The $x$-axis is logarithmic and truncated at $r\le 50$ for readability, and the vertical scale is shared across columns.}
    \label{fig:redundancy-vs-speedup-log50}
\end{figure}

In Figure~\ref{fig:redundancy-vs-speedup-log50}, the rolling-median \emph{Speedup} increases with redundancy on $r\in(0,50]$ for all except small $n$ ($n\neq 1000$), aligning with the intuition that higher redundancy lets the sparsifier safely discard a larger fraction of items and thereby reduces the effective search space of the original problem.

For $n=1000$, the tail bends downward. We interpret this as follows: when $r$ is large with small $n$, instances tend to contain a higher fraction of items that are easy for a modern ILP solver to eliminate through presolve or simple dominance reasoning. This reduces the effective candidate set in the full ILP, often bringing it close to the size of the sparsified instance. As a result, $T_{\mathrm{full}}$ and $T_{\mathrm{sparse}}$ naturally converge, which explains why the speedup ratio approaches~$1$ in the high-$r$ tail for $n{=}1000$.

From \(n{=}1000\) to \(n{=}5000\), the overall \emph{Speedup} level increases, consistent with the view that larger instances admit greater structural reduction after sparsification. At \(n{=}10000\), the median curve falls below the \(n{=}5000\) curve; however, the scatter reveals two strata: a high-speedup stratum and a lower, denser stratum. A median computed on the upper stratum alone would exceed the \(n{=}5000\) median and align with the otherwise monotone trend, whereas the concentration of points in the lower stratum dominates the rolling median and induces the observed downward deviation. To investigate this split, we next examine ETR slices by \(m\) and \(\rho\).

\paragraph{Equal-time value across redundancy.}
Figures~\ref{fig:etr-n1000}--\ref{fig:etr-n10000} plot \(\mathrm{ETR}=\OPT_{\mathrm{sparse}}/\OPT_{\mathrm{full}}^{(\mathrm{cut})}\) versus redundancy, stratified by \(m\) and \(\rho\). Excluding the \(n{=}10000,\,m{=}5\) anomaly discussed below, three consistent patterns emerge:
\begin{enumerate}
    \item \textbf{Scaling in \(n\) and \(r\).} Holding \((m,r,\rho)\) fixed, ETR increases with \(n\); holding \((n,m,\rho)\) fixed (and excluding \(m{=}1\)), ETR increases with \(r\). This is consistent with the view that larger instances and higher redundancy permit more effective, value-preserving sparsification of the instance.

    \item \textbf{Effect of \(m\).} When \(m{=}1\), mean ETR stays within \(0.99\)–\(1.00\) across \(r\) and \(\rho\). Our interpretation is that these slices are easy: the full ILP progresses quickly, the LP-driven preprocessing dominates \(T_{\mathrm{sparse}}\), and the equal-time baseline on the full instance nearly attains \(\OPT_{\mathrm{full}}\); consequently, ETR mirrors the \emph{Approximation ratio}.

    For \(m{=}2\), ETR is highest and exhibits the expected monotonicities: holding \((m,r,\rho)\) fixed, ETR increases with \(n\); holding \((n,m,\rho)\) fixed (excluding \(m{=}1\)), ETR increases with \(r\).
    
    For \(m{=}5\), ETR is uniformly below the \(m{=}2\) levels and the monotone patterns in \(n\) and \(r\) are less pronounced. However, at \(n{=}2000\) and \(n{=}5000\) a generally increasing trend with \(r\) remains visible. For the tail-side decline at high redundancy (e.g., \(r\) near \(50\)), we posit that -- under the definition \(r=n/|\OPT_{\mathrm{full}}|\) -- increasing \(m\) while holding \((n,r)\) fixed reduces the average number of selected items per knapsack, yielding sparser per-bin solutions and an effectively smaller search space for the full ILP. This accelerates \emph{exact solving} and narrows the gap to \(T_{\mathrm{sparse}}\), so the equal-time baseline more often attains values close to \(\OPT_{\mathrm{full}}\), leading to a downward drift in ETR as \(r\) approaches \(50\).

    \item \textbf{Effect of \(\rho\).} Fix \((n,m,r)\). The dependence on value–weight correlation exhibits two regimes.

    For \(m{=}1\), performance slightly degrades as \(\rho \to -1\).
    On these slices the ETR essentially coincides with the \emph{Approximation ratio}, so the observed drop reflects a small loss in \(\OPT_{\mathrm{sparse}}\) relative to \(\OPT_{\mathrm{full}}\) rather than a timing artifact.
    A plausible mechanism is that strong negative value–weight correlation induces a very sharp \(v/w\) ordering and a nearly integral single-knapsack optimum; with our deliberately coarse bucketing (\(\epsilon=0.2\)) and nominal per-bucket allowance (\(\tau=1\)), sparsification can exclude a few top candidates, and the steep tail then magnifies the impact of any such omission into a sub-percent objective gap.    
    
    For \(m\in\{2,5\}\), the pattern reverses: ETR improves as \(\rho\) becomes more negative. 
    Negative correlation concentrates the high-density items into the sparsifier’s buckets, so sparsification preferentially retains items that are likely to appear in \(\OPT_{\mathrm{full}}\) and filters out those that are not.
    Under the same time budget \(T_{\mathrm{sparse}}\), the time-constrained baseline on the full instance tends to find strong incumbents when \(\rho \to +1\), because many near-ties make near-optimal solutions easy to locate. 
    When \(\rho \to -1\), value dispersion makes high-quality incumbents harder to obtain within the time limit, so the baseline lags behind while the restricted ILP remains focused on the most relevant portion of the instance. 
    This gap is most pronounced at larger redundancy \(r\), where \emph{Speedup} is higher and the equal-time budget for the full instance is correspondingly tighter. 
    As a result, \(\mathrm{ETR}\) increases as \(\rho\) moves toward \(-1\) for \(m\in\{2,5\}\), highlighting that sparsification followed by exact solving identifies near-optimal solutions within the same wall-clock budget by quickly focusing on the most relevant items.
\end{enumerate}

\paragraph{Anomalous slice at \(n{=}10000,\,m{=}5\).}
For \(n{=}10000\) with \(m{=}5\), ETR concentrates near~\(1\) across the \((r,\rho)\) grid (Figure~\ref{fig:etr-n10000}). 
A run–level breakdown shows a systematic shift in wall–clock times: \(T_{\mathrm{full}}\) is \emph{smaller} than at \(n{=}5000,\,m{=}5\) under matched \((r,\rho)\), and it is also smaller than at \(n{=}10000,\,m{=}2\). 
These observations indicate that, in this slice, the full ILP becomes easier as \(n\) and \(m\) increase. 
As \(T_{\mathrm{full}}\) decreases, the time required for \emph{time-constrained exact solving} to perform well likewise shrinks; under a comparable budget \(T_{\mathrm{sparse}}\), the equal-time baseline frequently reaches (or nearly reaches) optimality.
Consistent with this, at \(n{=}5000\) we typically have \(T_{\mathrm{sparse}}\!\ll T_{\mathrm{full}}\), whereas at \(n{=}10000\) these two times are much closer -- reflected in Figure~\ref{fig:redundancy-vs-speedup-log50} by a dense lower–speedup stratum (containing both \(m{=}1\) and \(m{=}5\) points) for \(n{=}10000\); by contrast, in the \(n{=}5000\) panel most \(m{=}5\) points lie in the upper stratum.

We attribute this behavior to a generator–specific effect rather than to the sparsifier. 
Our data generator (Uniform/TruncNormal marginals coupled by a Gaussian copula at fixed \(\rho\); see Section~\ref{app:data}) becomes increasingly faithful to its target distributions as both \(n\) and \(m\) grow. 
With larger samples, empirical quantiles stabilize, the capacity rule becomes less noisy, and the induced instances exhibit more regular structure (e.g., clearer dominance relations among variables and fewer near–ties). 
Modern MILP presolve and cutting planes exploit such regularity effectively, compressing the root and closing the gap quickly; consequently, \(T_{\mathrm{full}}\) decreases while the LP–based preprocessing cost remains comparatively stable, which drives ETR toward~\(1\) in this slice. 
Because this phenomenon originates in the data generator rather than in the sparsification mechanism, we report the \(n{=}10000,\,m{=}5\) results for completeness (Figure~\ref{fig:etr-n10000}) but do not base our main claims on them.

\paragraph{Summary and takeaways.}

Taken together, the evidence supports our main thesis: in the deterministic regime, applying \emph{sparsification before exact solving} reduces effective instance size and yields wall--clock speedups while preserving nearly all objective value when structural redundancy is present. The gains are most pronounced in nontrivial, diagnostically informative slices---moderate--to--large redundancy \(r\), multiple bins \(m>1\), and larger \(n\) (up to \(5000\))---where the full ILP is genuinely time--consuming and sparsification concentrates the search on the most relevant part of the instance. By contrast, in slices with limited opportunity for sparsification---e.g., \(m{=}1\), very small \(n\), or very high \(r\) relative to \(n\)---the full ILP already progresses rapidly and the LP--based preprocessing constitutes a larger share of \(T_{\mathrm{sparse}}\); in these cases additional speedups are limited.


We also acknowledge a limitation of our synthetic pipeline (Section~\ref{app:data}): generating large batches of very large instances (e.g., \(n{=}10000\)) while simultaneously controlling difficulty is nontrivial. In particular, as both \(n\) and \(m\) grow, our generator becomes highly faithful to its target marginals and copula.

Overall, the results indicate that our sparsifier is most useful precisely where there is substantive scope for sparsification -- regimes with nontrivial structural redundancy and large feasible region, where \emph{exact solving} alone is computationally expensive but a carefully reduced instance incurs only negligible loss in objective value. An interesting direction for future work is to design richer generators or controlled real–world benchmarks with tunable redundancy, and to explore adaptive instantiations of our scheme (e.g., choosing \(\epsilon\) or bypassing the LP stage) that automatically recognize low-redundancy settings and concentrate effort where sparsification offers the greatest benefit.


\newcommand{\ETRConventions}{%
Columns are $m\in\{1,2,5\}$; rows 1–7 are per-$\rho$ violin plots with medians (red) and means (orange), and the bottom row is a $\rho$-colored scatter. The $x$-axis is logarithmic and truncated at $r\le 50$ for readability; vertical lines mark equal-width bins in $\log_{10}$, with per-bin sample sizes annotated above each panel. For $m=1$ the $y$-axis is fixed near $1$; for $m\in\{2,5\}$ we clip outliers at the $0.99$ quantile and unify $y$-limits within the figure to enable cross-panel comparison.%
}

\begin{figure}[t]
  \centering
  \includegraphics[width=\linewidth]{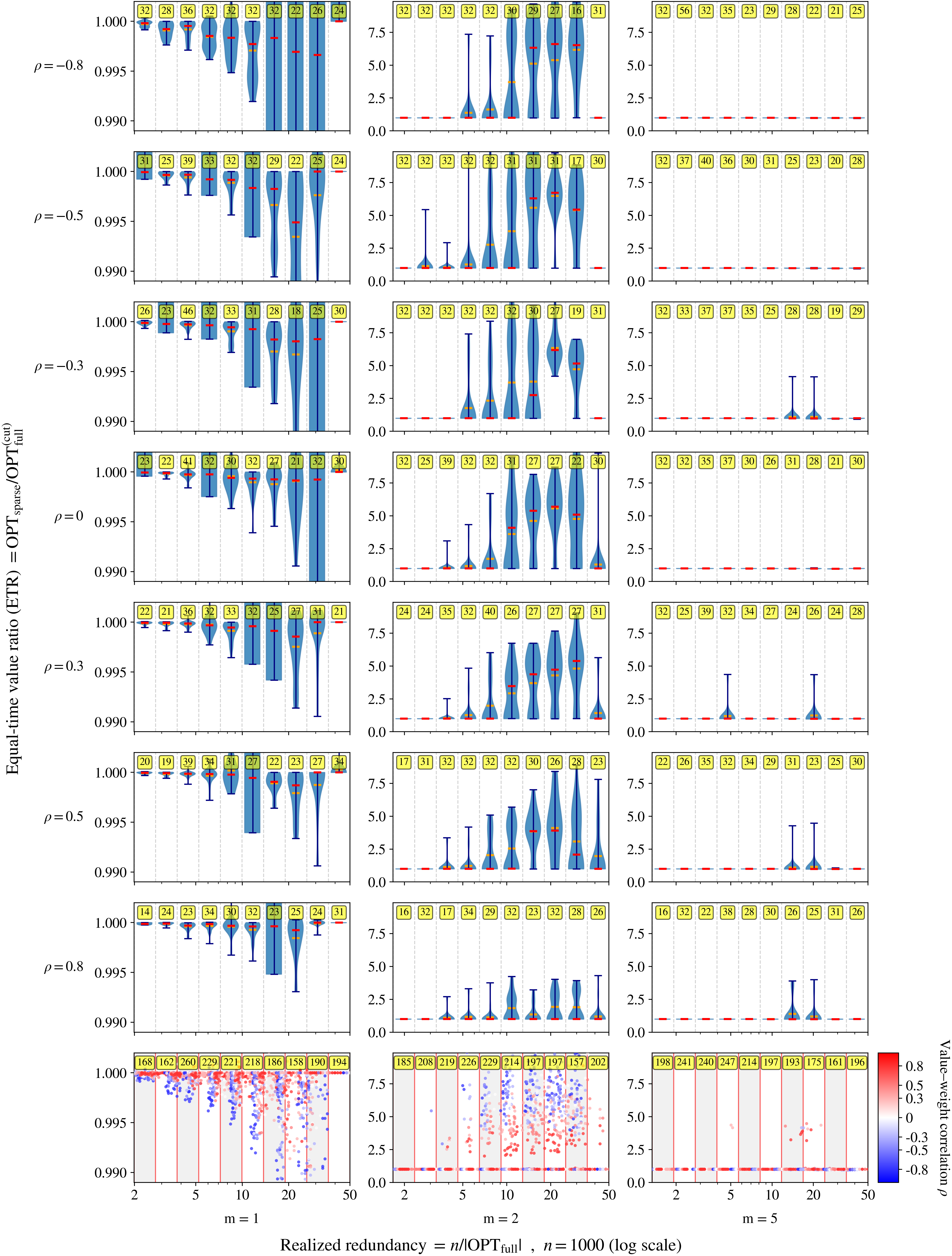}
  \caption{Equal-time value ratio (ETR $=\OPT_{\text{sparse}}/\OPT_{\text{full}}^{(\text{cut})}$) versus realized redundancy $r=n/|\OPT_{\text{full}}|$ for $n=1000$. \ETRConventions}
  \label{fig:etr-n1000}
\end{figure}

\begin{figure}[t]
  \centering
  \includegraphics[width=\linewidth]{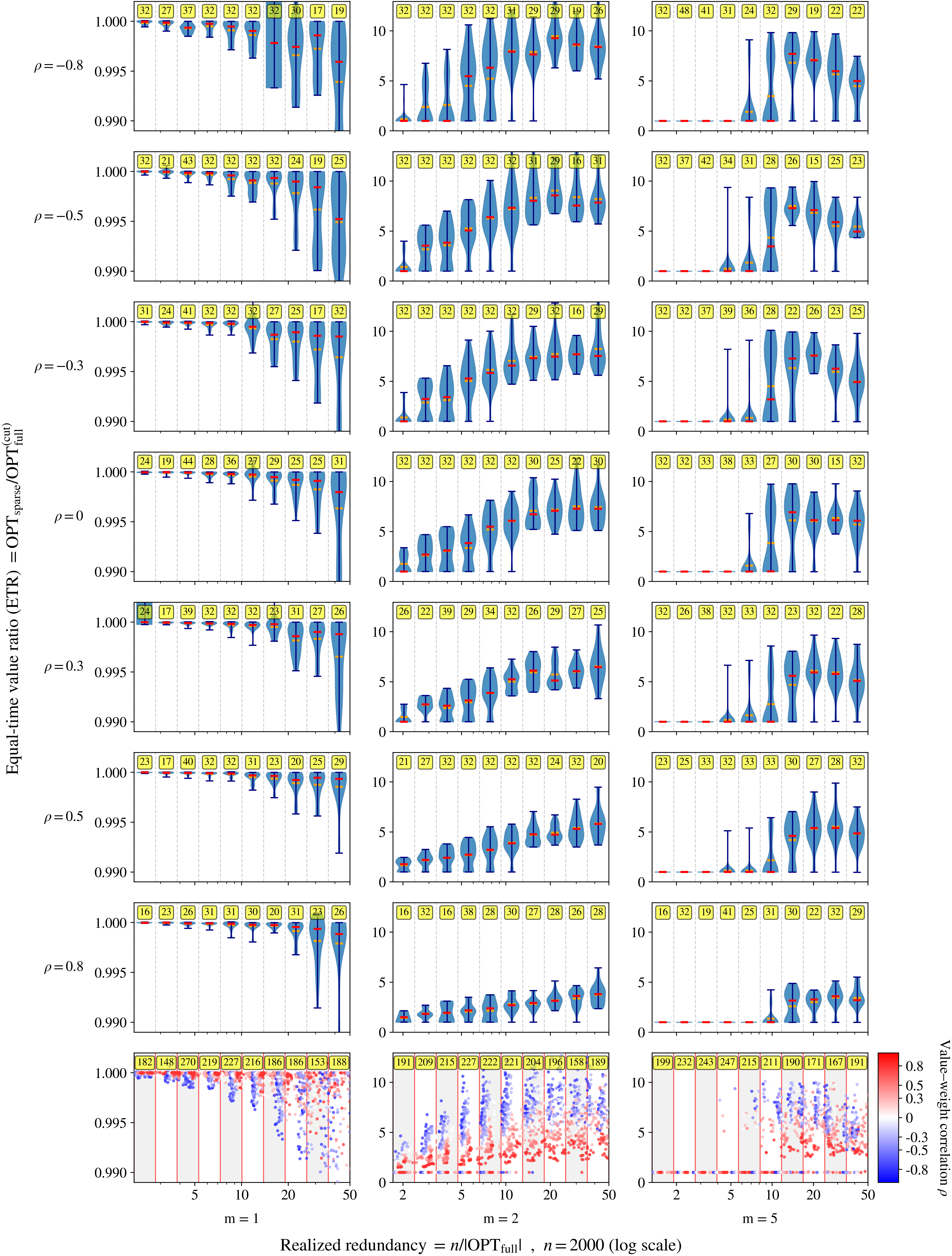}
  \caption{ETR versus realized redundancy for $n=2000$. Conventions as in Fig.~\ref{fig:etr-n1000}.}
  \label{fig:etr-n2000}
\end{figure}

\begin{figure}[t]
  \centering
  \includegraphics[width=\linewidth]{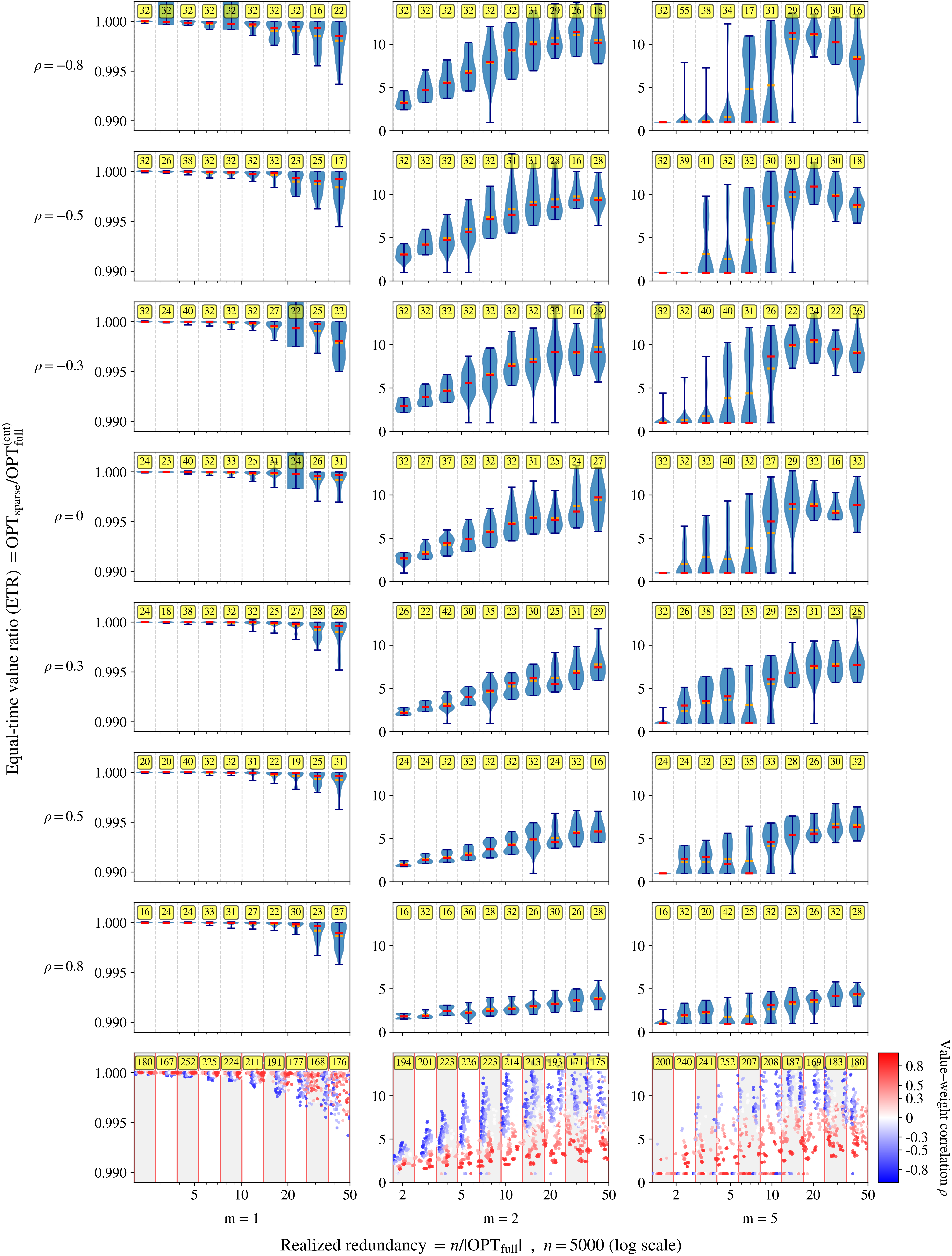}
  \caption{ETR versus realized redundancy for $n=5000$. Conventions as in Fig.~\ref{fig:etr-n1000}.}
  \label{fig:etr-n5000}
\end{figure}

\begin{figure}[t]
  \centering
  \includegraphics[width=\linewidth]{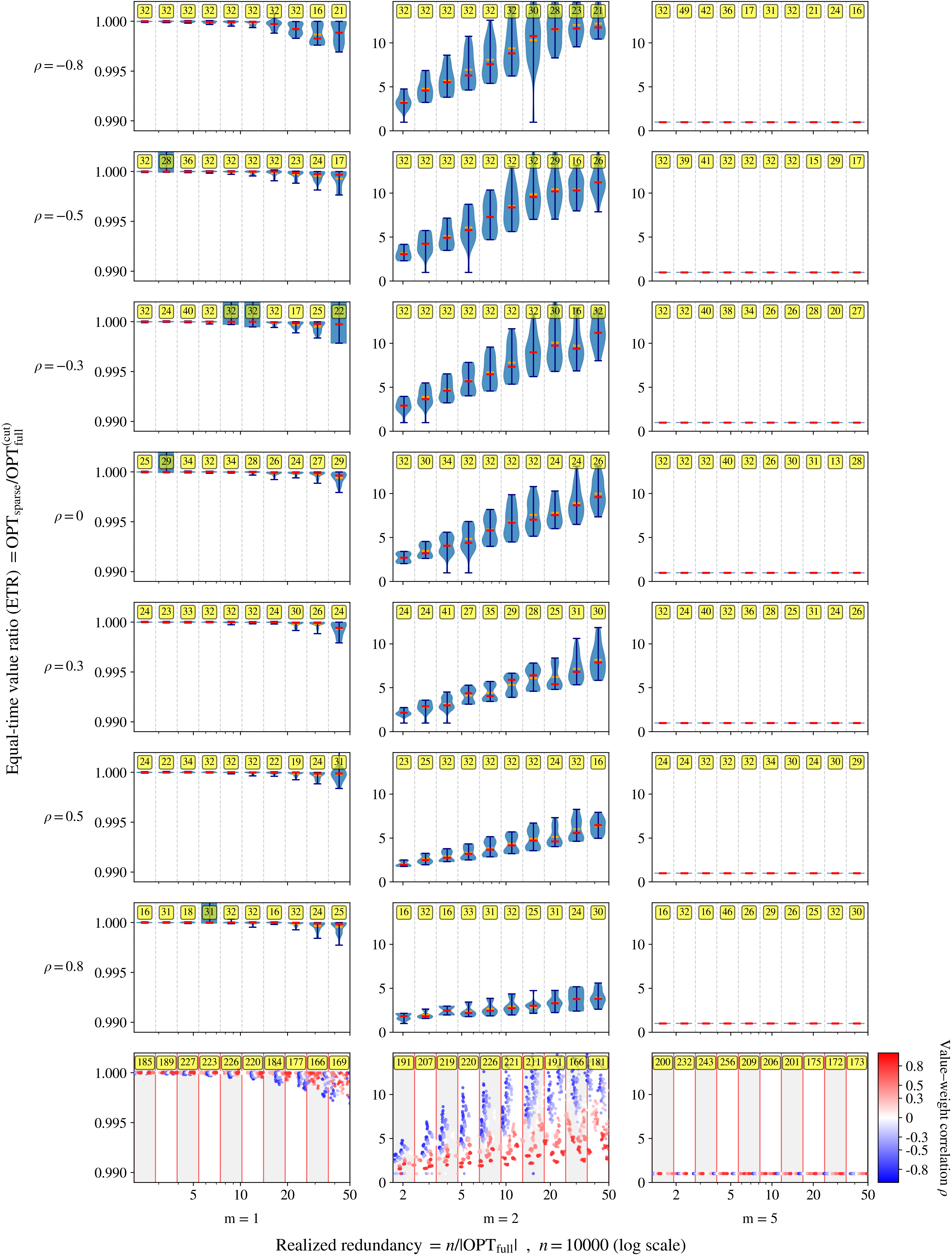}
  \caption{ETR versus realized redundancy for $n=10000$. Conventions as in Fig.~\ref{fig:etr-n1000}.}
  \label{fig:etr-n10000}
\end{figure}

\addcontentsline{toc}{section}{Appendix}

\end{document}